\documentclass[11pt]{article}
\usepackage[margin=1in]{geometry}
\usepackage[utf8]{inputenc}
\usepackage[T1]{fontenc}
\usepackage{authblk}
\usepackage{todonotes}
\usepackage{xspace}

\usepackage{comment}

\usepackage{nicefrac}
\usepackage{algorithm}
\usepackage{algorithmicx}
\usepackage[noend]{algpseudocode}

\usepackage{amssymb}
\usepackage{amsmath}
\usepackage{amsthm}
\usepackage{dsfont}
\usepackage{footnote}

\usepackage{xcolor}
\usepackage{nameref}
\definecolor{ForestGreen}{rgb}{0.1333,0.5451,0.1333}
\definecolor{DarkRed}{rgb}{0.65,0,0}
\definecolor{Red}{rgb}{1,0,0}
\usepackage[linktocpage=true,
pagebackref=true,colorlinks,
linkcolor=DarkRed,citecolor=ForestGreen,
bookmarks,bookmarksopen,bookmarksnumbered]
{hyperref}

\usepackage{cleveref}

\usepackage{thm-restate}
\usepackage[tableposition=bottom]{caption}
\usepackage[all]{hypcap}
\usepackage{subcaption}
\usepackage{framed}
\usepackage[framemethod=tikz]{mdframed}
\usepackage[skins]{tcolorbox}

\makeatletter
\g@addto@macro{\maketitle}{\@thanks}
\makeatother

\newtheorem{thm}{Theorem}[section]

\newtheorem{prop}[thm]{Proposition}

\newtheorem{fact}[thm]{Fact}
\newtheorem{lem}[thm]{Lemma}
\newtheorem{Def}[thm]{Definition}
\newtheorem{obs}[thm]{Observation}
\newtheorem{claim}[thm]{Claim}

\newtheorem{remark}[thm]{Remark}

\newcommand{\E}{\mathbb{E}}%

\newcommand{\eps}{\epsilon}%
\newcommand{\p}{\textsc{P}}%
\newcommand{\poly}{\mathrm{poly}}

\renewcommand{\algorithmiccomment}[1]{\bgroup\hfill$\rhd$~#1\egroup}

\newcounter{note}[section]

\newcommand{\calA}{\mathcal{A}}
\newcommand{\calC}{\mathcal{C}}

\newcommand{\calP}{\mathcal{P}}

\newcommand{\Otil}{\tilde{O}}

\newcommand{\Uni}{\textrm{Uni}}

\newenvironment{wrapper}[1]
{
	\begin{center}
		\begin{minipage}{\linewidth}
			\begin{mdframed}[hidealllines=true, backgroundcolor=gray!20, leftmargin=0cm,innerleftmargin=0.4cm,innerrightmargin=0.4cm,innertopmargin=0.4cm,innerbottommargin=0.4cm,roundcorner=10pt]
				#1}
			{\end{mdframed}
		\end{minipage}
	\end{center}
} 

\renewcommand{\paragraph}[1]{\medskip\noindent\textbf{#1}}

\usepackage{etoolbox}
\usepackage{tikz}
\usetikzlibrary{tikzmark}
\usetikzlibrary{calc}

\errorcontextlines\maxdimen

\newcommand{\ALGtikzmarkcolor}{black}
\newcommand{\ALGtikzmarkextraindent}{4pt}
\newcommand{\ALGtikzmarkverticaloffsetstart}{-.5ex}
\newcommand{\ALGtikzmarkverticaloffsetend}{-.5ex}
\makeatletter
\newcounter{ALG@tikzmark@tempcnta}

\newcommand\ALG@tikzmark@start{%
	\global\let\ALG@tikzmark@last\ALG@tikzmark@starttext%
	\expandafter\edef\csname ALG@tikzmark@\theALG@nested\endcsname{\theALG@tikzmark@tempcnta}%
	\tikzmark{ALG@tikzmark@start@\csname ALG@tikzmark@\theALG@nested\endcsname}%
	\addtocounter{ALG@tikzmark@tempcnta}{1}%
}

\def\ALG@tikzmark@starttext{start}
\newcommand\ALG@tikzmark@end{%
	\ifx\ALG@tikzmark@last\ALG@tikzmark@starttext
	\else
	\tikzmark{ALG@tikzmark@end@\csname ALG@tikzmark@\theALG@nested\endcsname}%
	\tikz[overlay,remember picture] \draw[\ALGtikzmarkcolor] let \p{S}=($(pic cs:ALG@tikzmark@start@\csname ALG@tikzmark@\theALG@nested\endcsname)+(\ALGtikzmarkextraindent,\ALGtikzmarkverticaloffsetstart)$), \p{E}=($(pic cs:ALG@tikzmark@end@\csname ALG@tikzmark@\theALG@nested\endcsname)+(\ALGtikzmarkextraindent,\ALGtikzmarkverticaloffsetend)$) in (\x{S},\y{S})--(\x{S},\y{E});%
	\fi
	\gdef\ALG@tikzmark@last{end}%
}

\apptocmd{\ALG@beginblock}{\ALG@tikzmark@start}{}{\errmessage{failed to patch}}
\pretocmd{\ALG@endblock}{\ALG@tikzmark@end}{}{\errmessage{failed to patch}}
\makeatother

\title{Dynamic Matching with Better-than-2 Approximation in Polylogarithmic Update Time}

 \author[1]{Sayan Bhattacharya\thanks{Supported by Engineering and Physical Sciences Research Council, UK (EPSRC) Grant EP/S03353X/1.}}
 \author[1]{Peter Kiss\thanks{Work done in part while the author was visiting Max-Planck-Institut für Informatik}}
 \author[2]{Thatchaphol Saranurak}
 \author[3]{David Wajc\thanks{Work done in part while the author was at Stanford University.}}
 \affil[1]{University of Warwick}
 \affil[2]{University of Michigan, Ann Arbor}
 \affil[3]{Google Research}

\begin{document}

\pagenumbering{gobble}

\maketitle

\begin{abstract}
We present dynamic algorithms with \emph{polylogarithmic} update time
for estimating the size of the maximum matching of a graph undergoing
edge insertions and deletions with approximation ratio \emph{strictly
better than $2$}. Specifically, we obtain a $1+\frac{1}{\sqrt{2}}+\epsilon\approx 1.707+\epsilon$
approximation in bipartite graphs and a $1.973+\epsilon$ approximation
in general graphs. We thus answer in the affirmative the value version of the major open question repeatedly asked in the dynamic graph algorithms literature. Our randomized algorithms' approximation and worst-case update time bounds both hold w.h.p.~against adaptive adversaries.

\smallskip

Our algorithms are based on simulating new two-pass streaming matching
algorithms in the dynamic setting. Our key new idea is to invoke the recent sublinear-time
matching algorithm of Behnezhad (FOCS'21) in a white-box manner to efficiently simulate the second pass of our streaming algorithms, while bypassing the well-known vertex-update barrier.
\end{abstract}

\newpage

\tableofcontents
\newpage 

\pagenumbering{arabic}
\section{Introduction}\label{sec:intro}	
	
	The maximum matching problem is a cornerstone of combinatorial optimization and theoretical computer science more broadly. (We recommend \cite{duan2014linear} for a brief history of this problem.)
	The study of this problem and its extensions has contributed foundational advances and concepts to the theory of computing, from the introduction of the primal-dual method \cite{kuhn1955hungarian}, impact on polyhedral combinatorics \cite{edmonds1965maximum}, and the advocacy for  polynomial-time computability as the measure of efficiency (in static settings) \cite{edmonds1965paths}.
	
	The maximum matching problem has also been intensely studied in \emph{dynamic} settings. 
	Here, the graph undergoes edge \emph{updates} (insertions and deletions), and we wish to approximate the maximum matching, while spending little computation time between updates, referred to as \emph{update time}.
	Polynomial update time is trivial to achieve by running exact static algorithms (e.g., \cite{edmonds1965paths}) after each update. However, intuitively, such minor changes to the graph should allow for much faster algorithms, with possibly even exponentially smaller, \emph{polylogarithmic} update times.

	The first sublinear (i.e., $o(m)=o(n^2)$) update time dynamic matching algorithm was given 15 years ago by Sankowski \cite{sankowski2007faster}, who used fast dynamic matrix inversion to maintain the maximum matching size in update time $O(n^{1.495})$, recently improved to $O(n^{1.407})$ \cite{van2019dynamic}.
	Unfortunately, a number of fine-grained complexity results rule out fast, and even sublinear-in-$n$ update time \cite{abboud2014popular,henzinger2015unifying,abboud2016popular,dahlgaard2016hardness,kopelowitz2016higher} for (exact) maximum matching size estimation.
	This motivates the wealth of work on computing \emph{approximate} matchings dynamically.
	
	The first polylogarithmic update time dynamic matching algorithm is due to an influential work of Onak and Rubinfeld \cite{onak2010maintaining}, who gave a (large) constant approximation in polylog update time.
	This was later improved by Baswana et al.~\cite{baswana2015fully} to a $2$-approximation in logarithmic update time, later improved to constant time by Solomon \cite{solomon2016fully}.
	Numerous other algorithms achieving a $2$- or $(2+\eps)$-approximation in polylog update time were subsequently developed, with expected amortized update time improved to worst-case w.h.p.,\footnote{An algorithm has \emph{amortized} update time $f(n)$ if every sequence of $t$ updates starting from an empty graph takes at most $t\cdot f(n)$ update time. If each operation takes at most $f(n)$ time, it has \emph{worst-case} update time $f(n)$.} and oblivious randomized algorithms improved to advsersarially-robust ones, and then to deterministic ones \cite{bhattacharya2016new,arar2018dynamic,charikar2018fully,bernstein2019deamortization,bhattacharya2019deterministically,behnezhad2019fully,chechik2019fully,wajc2020rounding,bhattacharya2021deterministic,kiss2022improving}.\footnote{An algorithm \emph{works against an adaptive adversary} if its guarantees hold even when future updates depend on the algorithm's previous output. We also say that such an algorithm is \emph{adversarially robust}, or \emph{robust} for short. 
	The importance of robustness for \emph{static} applications has motivated a recent concentrated effort to design robust dynamic algorithms for myriad problems (see, e.g., discussions in \cite{nanongkai2017dynamic,beimel2022dynamic,wajc2020rounding,bhattacharya2021deterministic,chuzhoy2019new,fahrbach2018graph}).
	}
	
	A complementary line of work studied better-than-two-approximate dynamic matching, providing a number of small polynomial (even sublinear in $n$) update times for approximation ratios below the natural bound of $2$ achieved by inclusionwise maximal matchings. This includes $(1+\eps)$-approximate algorithms with $O_\eps(\sqrt{m})=O_{\eps}(n)$ update time \cite{gupta2013fully,peleg2016dynamic}, $\left(\frac{3}{2}+\eps\right)$-approximate algorithms with $O_\eps(\sqrt[4]{m})=O_{\eps}(\sqrt{n})$ update time \cite{bernstein2015fully,bernstein2016faster,grandoni2022maintaining,kiss2022improving} and a number of tradeoffs between approximation in the range $(3/2,2)$ and sublinear-in-$n$ polynomial update times \cite{behnezhad2020fully,wajc2020rounding,bhattacharya2021deterministic,behnezhad2022new,roghani2022beating}.\footnote{Throughout the paper, we use $O_\eps(\cdot)$ to suppress $\poly(1/\eps)$ factors and $\Otil(\cdot)$ to suppress $\poly(\log n)$ factors.}

	This state of affairs leaves open a key question, repeatedly raised in the literature \cite{bhattacharya2016new,bernstein2016faster,charikar2018fully,behnezhad2020fully,wajc2020rounding,behnezhad2022new,le2022dynamic} and first posed by Onak and Rubinfeld in their aforementioned groundbreaking work \cite{onak2010maintaining}:
	\begin{quote}
		\emph{How small can [approximation factors] be made with polylogarithmic update time? [...] Can the approximation constant be made smaller than 2 for maximum matching?}\\
	\end{quote}

	\subsection{Our Results}\label{sec:results}
	
	We resolve the question of polylogarithmic update time better-than-two-approximate dynamic matching algorithms in the affirmative, \emph{for the value version of the problem}. That is, letting $\mu(G)$ denote the maximum matching size in $G$, we maintain an estimate $\nu$ that is $\alpha<2$ approximate, i.e., it satisfies
	$\nu \leq \mu(G)\leq \alpha\cdot \nu$ at every point in time.
	Our main result is the following.
	
	\begin{wrapper}
		\begin{restatable}{thm}{maintheorem}\label{main-theorem}
		For every $\eps\in (0,1)$, there exists a randomized $(1.973+\eps)$-approximate dynamic matching size estimation algorithm with $\poly(\log n,1/\eps)$ worst-case update time.
		Both the algorithm's approximation ratio and update time hold w.h.p against an adaptive adversary.
		\end{restatable}
	\end{wrapper}
	
	For bipartite graphs, we obtain a stronger approximation guarantee of $1+\frac{1}{\sqrt{2}}+\eps\approx 1.707+\eps$.

	\paragraph{Secondary results.}
	Our approach is versatile, and yields the following generic reduction.
	\begin{restatable}{thm}{tradeoffs}\label{thm:tradeoffs}
	For any $\alpha>1.5$, a dynamic $\alpha$-approximate matching algorithm with update time $t_u$ implies a dynamic $\left(\alpha - \Omega \left(\left(1-6\left(\frac{1}{\alpha}-\frac{1}{2}\right)\right)^2\right)\right)$-approximate matching \emph{size estimator} with update time $\Otil(t_u)$.
	\end{restatable}
	Very recently, Behnezhad and Khanna \cite{behnezhad2022new} presented new dynamic matching algorithms trading off approximations $\alpha\in (1.5,2]$ and small polynomial update times. Applying \Cref{thm:tradeoffs} to their algorithms, we obtain improved approximation for dynamic matching size estimation,
	within the same update time up to polylog factors.

	To obtain our main results, we design several $2$-pass \emph{semi-streaming} algorithms	(see \Cref{sec:related}), including a deterministic $(1+1/\sqrt{2}+\eps)$-approximate algorithm on bipartite graphs. This matches the prior state-of-the-art \cite{konrad2021two,konrad2018simple} up to an $\eps$ term, while removing the need for randomization.

	\subsection{Our Techniques}\label{sec:techniques}
	
    We take the following high-level approach to prove \Cref{main-theorem}: (1) compute a maximal (and hence $2$-approximate) matching $M_1$, and  (2) augment $M_1$ if it is no better than $2$-approximate, using the myriad short augmenting paths $M_1$ must have in this case. This approach is common in many computational models, including the 2-pass semi-streaming model (see \Cref{sec:related}). Implementing this approach in a dynamic setting, however, faces several challenges. 
    The first challenge if we want robust algorithms with low worst-case update times is that no robust (near-)maximal matching algorithms with worst-case $\tilde{O}_{\epsilon}(1)$ update time are known.
    Of possible independent interest, we resolve this first challenge in \Cref{sec:adaptive}, by leveraging the robust fast matching sparsifiers of \cite{wajc2020rounding}.
	
	The more central challenge when trying to implement the above approach is that the search for augmenting paths  requires us to find many (disjoint) edges between matched and unmatched nodes in $M_1$. In a (multi-pass) streaming setting, this can be done by computing a large ($b$-)matching in the bipartite graph induced by edges in $V(M_1)\times \overline{V(M_1)}$. In a dynamic setting, however, this requires us to deal with \emph{vertex updates}, which are notoriously challenging in the context of dynamic matching, and all algorithms to date require reading all $\Omega(n)$ edges of each updated vertex \cite{le2022dynamic}.
	
	To overcome the above key challenge, we first note that we do not need to handle vertex updates individually, but may instead process these in \emph{batches} of $\Theta(\eps n)$ vertex updates, building on the periodic recomputation and sparsification techniques common in the literature (see \Cref{reduction}). Our main observation is that these batches of vertex updates, which  need to be handled if we wish to maintain the $b$-matchings from the second pass of our semi-streaming algorithms, can be implemented in $\tilde{O}_{\eps}(n)$ time using the sublinear-time algorithm of Behnezhad \cite{behnezhad2022time}. This leads to an amortized ${\tilde{O}_{\epsilon}(n)}/({\epsilon n}) = \tilde{O}_{\eps}(1)$ additive overhead in the update time (easily deamortized), implying our main result. This approach is versatile, and similarly underlies our secondary results.

	\subsection{Further Related Work}\label{sec:related}
	Having discussed the rich literature on the dynamic matching problem above, we do not elaborate on it further here. We do, however, highlight some connections to the literature on matching in other computational models that is closely related to our work.
	
	\paragraph{Streaming Matching.} 
	In the (semi-)streaming model, an $n$-node graph is revealed in a stream, edge by edge, and we wish to compute a large matching using only (optimal) $\tilde{O}(n)$ space. A line of work studying the problem of computing an approximately-maximum weighted matching \cite{feigenbaum2005graph,mcgregor2005finding,epstein2013improved,crouch2014improved,paz20182+,ghaffari2019simplified} culminated in a $\left(2+\eps\right)$-approximation \cite{paz20182+,ghaffari2019simplified}.  
	For unweighted graphs, lower bounds are known \cite{goel2012communication,kapralov2013better,kapralov2021space},
	but it remains a major open question whether one can break the barrier of $2$-approximation achievable by a trivial maximal matching algorithm.
	Striving for better approximation (and insights to break this barrier), several works designed algorithms using \emph{multiple passes} over the stream \cite{mcgregor2005finding,kale2017maximum,feldman2021maximum,esfandiari2016finding,ahn2013linear,eggert2009bipartite,fischer2022deterministic,eggert2009bipartite,ahn2013linear,assadi2022semi}.
	For $2$ passes, the state-of-the-art approximation ratios are $1.857$ \cite{feldman2021maximum}, and $1+\frac{1}{\sqrt{2}}\approx 1.707$ for bipartite graphs using the randomized algorithms of \cite{konrad2018simple,konrad2021two}, with the best prior deterministic bound being $\frac{12}{7}\approx 1.714$ \cite{esfandiari2016finding}.
	
	\paragraph{Sublinear-Time Matching.} 
	Computation of large matchings in sublinear \emph{time} has also been the subject of great interest. In regular bipartite graphs, a maximum matching can be computed in $\tilde{O}(n)$ time \cite{goel2009perfect,goel2010perfect,goel2013perfect}. In general  graphs with bounded-degrees, it was known how to achieve a $(2+\epsilon)$-approximation in sublinear time \cite{nguyen2008constant,onak2012near,parnas2007approximating,yoshida2012improved}. This was recently improved to a $\tilde{O}(n)$ time algorithm for {\em any} general graph \cite{behnezhad2022time}. As discussed in \Cref{sec:techniques}, we use this latter algorithm in a white-box manner to obtain our main result.

\subsubsection{Concurrent work} Independently and concurrently, Behnezhad \cite{behnezhad2023dynamic} (in a work in the same conference) obtained the same main qualitative result as ours: a better-than-two-approximate polylogarithmic time dynamic matching size estimation algorithm. 
The basic approach to achieve this qualitative result is the same in both papers: Simulate the second pass of a two-pass streaming algorithm using the sublinear-time algorithm of \cite{behnezhad2022time}, together with batched computation. The quantitative differences in the papers' approximation ratios are due to the two-pass streaming algorithms used---our new maximal-b-matching-based algorithms here, and an algorithm inspired by \cite{konrad2012maximum} in \cite{behnezhad2023dynamic}. We note that \cite{behnezhad2023dynamic} also achieves $(3/2-\Omega(1))$-approximate size estimation algorithm in time ${O}(\sqrt{n})$ (the best update times for $(3/2+\eps)$-approximate explicit matching \cite{bernstein2015fully,bernstein2016faster,grandoni2022maintaining,kiss2022improving}). This result also uses the high-level approach of batched computation using sublinear-time algorithms, building on a new characterization of tight examples for the $3/2$-approximate matching sparsifiers (EDCS) of \cite{bernstein2015fully}.


\color{black}
\vspace{-0.05cm}
	\section{Preliminaries}\label{sec:prelims}
\vspace{-0.05cm}	
	Our input is a graph $G$ on $n$ nodes $V$, with an initially empty edge set $E$, undergoing edge updates (insertions and deletions).
	Our objective is to approximate the maximum matching size $\mu(G)$ well, while spending little update time (computation between updates).
	In addition, we want our algorithms to work in the strictest settings: against an adaptive adversary (i.e., their guarantees hold for any update sequence), and with small \emph{worst-case} update time guarantees.

	\paragraph{Matching theory basics.} A \emph{matching} is a vertex-disjoint subset of edges. A \emph{maximal matching} is an inclusionwise-maximal matching. 
	A maximum matching is a matching of largest cardinality. 
	In a weighted graph with edge weights $w_e\in \mathbb{R}$, a maximum weight matching is a matching $M$ of largest total weight, $w(M):=\sum_{e\in M} w_e$.
	An \emph{augmenting path} $P$ with respect to a matching $M$ is a simple path starting and ending with distinct nodes unmatched in $M$, with the edges alternatingly outside and inside $M$. 
	Setting $M\gets M\bigoplus P$, where $\bigoplus$ denotes the symmetric difference, referred to as \emph{augmenting} $M$ along $P$, increases the cardinality of $M$ by one.
	A \emph{$b$-matching} with capacities $\{b_v\}_{v\in V}$ is a collection of \emph{multi-edges} $F$ of $E$ (that is, edges of $E$ may appear multiple times in $F$) with no vertex $v$ having more than $b_v$ multi-edges in $F$.
	A \emph{fractional matching} $x:E\to \mathbb{R}_{\geq 0}$ assigns non-negative values to edges so that each vertex $v$ has \emph{fractional degree} $\sum_{e\ni v} x_e$ at most one.
	In bipartite graphs, the existence of a fractional matching of size $k$ implies the existence of an integral matching of cardinality $\lceil k\rceil$. In general graphs, 
	this fractional relaxation has a maximum integrality gap of $3/2$, attained by a triangle graph with values $x_e=1/2$ for each edge $e$.
	
	\paragraph{Notation:}
	Let $V(M)$ denote the set of all endpoints of edges in a matching $M$, and let $\overline{V(M)} := V\setminus V(M)$. For any disjoint vertex sets $A,B \subseteq V$, we let $G[A,B]$ denote the bipartite subgraph induced by the edges in $G$ with one endpoint in $A$ and another in $B$. Finally, for any subset of edges $E' \subseteq E$, we let $G[E']$ denote the subgraph of $G$ induced by $E'$.
	
	\subsection{Previous building blocks}

    A ubiquitous paradigm in the approximate dynamic matching literature is \emph{periodic recomputation}, introduced by Gupta and Peng \cite{gupta2013fully}.
	This approach is particularly useful in conjunction with sparsification techniques. We will use the vertex sparsification technique introduced by Assadi et al.~\cite{assadi2019stochastic} in the context of stochastic optimization, and adapted to dynamic settings by Kiss \cite{kiss2022improving}. 
	Combined, these approaches yield the following ``reduction'' from dynamic matching algorithms with \emph{immediate} queries to ones with slower query time.

	\begin{restatable}{prop}{reduction}\label{reduction}
	Let $\eps\in (0,1)$ and $\alpha \geq 1$.
    Suppose there exists an algorithm $\calA$ on a dynamic $n$-node graph $G$ with update time $t_u$, that, 
    provided $\mu(G)\geq \eps\cdot n$, supports $t_q$-time $\alpha$-approximate size estimate queries w.h.p.
    Then, there is another algorithm $\calA'$
    on $G$ that always maintains an $(\alpha+O(\eps))$-approximate estimate $\nu'$ in $\tilde{O}_{\epsilon}(t_u+t_q/n)$ update time. Moreover if the update time of $\calA$ is worst-case, so is that of $\calA'$, and if $\calA$ works against an adaptive adversary, then so does $\calA'$.
    \end{restatable}
    
	The above proposition, implicit in prior work, 
	serves as a useful abstraction, and so we provide a proof of this proposition for completeness in \Cref{appendix:previous-blocks}.
    As discussed in \Cref{sec:techniques}, this reduction is one of the crucial ingredients that allows us to bypass the vertex-update barrier.

	Another key ingredient we use is the sublinear-time (approximate) maximal matching algorithm of Benhezhad \cite{behnezhad2022time}, whose guarantees are captured by the following proposition (see \Cref{appendix:previous-blocks}).
	
	\begin{prop}\label{sublinear}
		Let $\eps\in (0,1/2)$. Using $\tilde{O}_{\epsilon}(n)$ time and $\tilde{O}_{\eps}(n)$ adjacency matrix queries w.h.p.~in an $n$-node graph $G$, one can compute a value $\nu$ which approximates $\tilde{\mu}$, the size of some maximal matching in $G$, within additive error $\eps n$. Namely,
		$\tilde{\mu} \geq \nu \geq \tilde{\mu} -\eps n.$
	\end{prop}
	
	A simple combination of propositions \ref{reduction} and \ref{sublinear} (with $t_q=\tilde{O}_{\eps}(n)$) immediately yields (yet) another $\tilde{O}_{\eps}(1)$-time $(2+\eps)$-approximation algorithm. As we will show, these propositions are also useful ingredients for breaking the barrier of $2$-approximation within the same update time.

	\subsection{New algorithmic primitive: Robust Approximately Maximal Matchings}
	
	To make our algorithms robust against adaptive adversaries we need an algorithm for maintaining \emph{approximately-maximal matchings} (AMM), which are defined as follows.
	
	\begin{Def}[\cite{peleg2016dynamic}]\label{def:eps-AMM}
		A matching $M$ is an \emph{$\eps$-approximately maximal matching ($\eps$-AMM)} in graph $G$ if $M$ is maximal in some subgraph obtained by removing at most $\eps\cdot \mu(G)$ nodes of $G$.
	\end{Def}
	\begin{obs}\label{AMM-size-bound}
	    If $M$ is an $\eps$-AMM in $G$, then $|M|\geq \frac{1}{2}(1-\eps)\cdot \mu(G) = \left(\frac{1}{2}-\frac{\eps}{2}\right)\cdot \mu(G)$.
	\end{obs}
		
	Peleg and Solomon \cite{peleg2016dynamic} showed how to maintain an $\eps$-AMM quickly in bounded-arboricity (i.e., globally sparse) graphs.
	In \Cref{sec:adaptive} we show how to maintain such matchings quickly in arbitrary graphs, proving the following.
	
	\begin{restatable}{lem}{epsAMM}\label{eps-AMM-algo}
        For any $\eps\in (0,1)$, there exists a robust dynamic algorithm that w.h.p.~maintains an $\eps$-AMM in worst-case update time $\tilde{O}_{\eps}(1)$.
	\end{restatable}

	A well-known fact is that a maximal matching that is close to $2$-approximate must admit many length-three augmenting paths (see e.g., \cite[Lemma 1]{konrad2012maximum}). 
	Our interest in AMMs is in part motivated by the following slight generalization of this fact.

	\begin{restatable}{prop}{shortaug}\label{poor-approx-maximal-matching=many-aug-paths}
		Let $\eps>0$ and $c \in \mathbb{R}$ and let $M$ be an $\eps$-AMM in $G$ such that $|M|\leq \left(\frac{1}{2}+c\right)\cdot \mu(G)$. Then $M$ admits a collection of at least $\left(\frac{1}{2}-3c-\frac{7\eps}{2}\right)\cdot \mu(G)$  node-disjoint 3-augmenting paths.
	\end{restatable}
	\section{Algorithms on Bipartite Graphs}
	
	In this section we illustrate our techniques for the special case of bipartite graphs, for which we obtain an improved approximation ratio of $1+\frac{1}{\sqrt{2}}+\eps \approx 1.707+\eps$.
	
	\subsection{Two-Pass Streaming Algorithm}
	Here we present our deterministic 2-pass streaming algorithm. We first compute an approximately-maximal matching $M_1$ from the first pass.\footnote{We suggest to the reader to think of $M_1$ as a maximal matching (i.e., $\eps=0$). We relax $M_1$ to be an $(\eps/8)$-AMM since this will be useful in our dynamic implementation that works against adaptive adversaries.} Then, in the second pass, we compute a maximal $b$-matching $M_2$ in the graph between matched and unmatched vertices, with capacities $k$ and $\lfloor k\cdot b\rfloor$, respectively, where we set the parameters $k \in \mathbb{Z}$ and $b \in \mathbb{R}$ later.
	Finally, we output an estimate $(1-\delta)\cdot |M_1|+(\delta/k)\cdot |M_2|$ where 
	$\delta = 1/b$.
	Our pseudocode is given in \Cref{alg:bipartite}.
	
	\begin{algorithm}[H]
		\caption{Bipartite Two-Pass Streaming Algorithm}
		\label{alg:bipartite}
		\begin{algorithmic}[1]
			\State $M_1\gets$ $(\eps/8)$-AMM in $G$ computed from first pass 
			\State assign each vertex $v$ capacity $b_v=\begin{cases}
			k & v\in V(M_1) \\
			\lfloor k\cdot b\rfloor & v\not\in V(M_1)
			\end{cases}$ 
			\State $M_2\gets $ maximal $b$-matching in $G[V(M_1),\overline{V(M_1)}]$  computed from second pass
			\State \textbf{Output} $(1-\delta)\cdot |M_1|+(\delta/k) \cdot |M_2|$.
		\end{algorithmic}
	\end{algorithm}	
	\vspace{-0.2cm}
	First, we prove that the above algorithm's output estimate corresponds to a matching in $G$.
	
	\begin{obs}\label{bip-bound}
		We have that $\mu(G[M_1\cup M_2])\geq (1-\delta)\cdot |M_1|+(\delta/k)\cdot |M_2|.$
	\end{obs}
	\begin{proof}
		Since $G$ is bipartite, by the integrality of the bipartite fractional matching polytope, to prove that $G':=G[M_1\cup M_2]$ contains a large matching witnessing the desired inequality, it suffices to prove that $G'$ contains a \emph{fractional} matching $\vec{x}$ of value $\sum_e x_e = (1-\delta)\cdot |M_1|+(\delta/k)\cdot |M_2|$. Indeed, such a fractional matching is obtained by assigning edge values
		$$x_e = \begin{cases} 1-\delta & e\in M_1 \\
		\delta/k & e\in M_2\setminus M_1.\end{cases}$$
		This is indeed a fractional matching, since each vertex $v$ has bounded fractional degree, $\sum_{e\ni v} x_e\leq 1$: every vertex $v\in V(M_1)$ has one incident $M_1$ edge and at most $k$ many incident $M_2$ edges, and so $\sum_{e\ni v}x_e\leq (1-\delta)+(\delta/k)\cdot k=1$, while every vertex $v\not\in V(M_1)$ has no incident $M_1$ edge and has at most $k\cdot b$ incident $M_2$ edges, and so $\sum_{e\ni v}x_e\leq k\cdot b\cdot (\delta/k) \leq 1$. 
	\end{proof}
	
	By \Cref{bip-bound}, \Cref{alg:bipartite} outputs a valid estimate for the matching size, $\nu\leq \mu(G)$. It remains to prove that $\nu$ provides a good approximation of $\mu(G)$. For this, we require the following.
	
	\begin{lem}\label{lm:maximal_matching_size}
		Let $M$ be a maximal $b$-matching in a bipartite graph $G = (L \cup R, E)$, with positive integral capacities $b_v = \ell$ for all $v \in L$ and $b_v = r$ for all $v \in R$. Then $$|M| \geq  \mu(G) \cdot \frac{\ell \cdot r}{\ell+r}.$$
	\end{lem}
	\begin{proof}
		Fix a maximum matching $M^*$ in $G$. Next, we define the subset of matched nodes in $M^*$ that are also saturated in $M$. That is, if $d_M(v)$ is $v$'s degree in $M$, we let
		\begin{align*}
		L^*_{sat} & := \{u \in L\cap V(M^*) \mid d_M(u) = \ell\}\\
		R^*_{sat} & := \{v \in R\cap V(M^*) \mid d_M(v) = r \}.
		\end{align*}
		Let $\alpha := |L^*_{sat}|/|M^*|$ and $\beta := |R^*_{sat}|/|M^*|$ denote the fraction of $M^*$ edges with a saturated $L$ and $R$ node, respectively. Since $M$ is a maximal  $b$-matching in $G$, each edge has at least one saturated endpoint, and so $\alpha + \beta\geq 1$. 
		By double counting the edges of $M$, relying on $\alpha + \beta\geq 1$, and noting that $\alpha\cdot r + (1-\alpha)\cdot \ell$ attains its minimum of $2\frac{\ell\cdot r}{\ell+r}$ at $\alpha=\frac{r}{\ell+r}$, we obtain the claimed inequality.
		\begin{align*}
		|M| & = \frac{1}{2} \left(\sum_{v \in L} d_M(v) +  \sum_{v \in R} d_M(v) \right)   \\
		& \geq \frac{1}{2} \left(\sum_{v \in L^*_{sat}} d_M(v) + \sum_{v \in R^*_{sat}} d_M(v) \right) \\
		& = \frac{1}{2} \cdot (|M^*| \cdot \alpha \cdot \ell + |M^*| \cdot \beta \cdot r)  \\
		& \geq \frac{1}{2} \cdot \mu(G)  \cdot (\alpha \cdot \ell + (1-\alpha) \cdot r) \\
		& \geq \mu(G)  \cdot \frac{\ell \cdot r}{\ell + r}. \qedhere
		\end{align*}
	\end{proof}
	
	We are now ready to bound the approximation ratio of \Cref{alg:bipartite}.
	
	\begin{lem}\label{bip:approx}
		For any $\eps\in (0,1)$, \Cref{alg:bipartite} with $b =1+\sqrt{2}$ and $k\geq \frac{8}{\eps b}$ run on bipartite graph $G$ computes a $(1+\frac{1}{\sqrt{2}}+\eps) \approx (1.707 +\epsilon)$-approximation to $\mu(G)$.
	\end{lem}
	\begin{proof}
		Fix a maximum matching $M^*$ in $G$. Next, for $i\in \{0,1,2\}$, let $M^*_i$ denote the edges of $M^*$ with $i$ endpoints matched in $M_1$. By definition, and since $|M_1|=\frac{1}{2}\cdot |V(M_1)|$, we have that 
		\begin{align}\label{eqn:M1-vs-M^*_i}
			|M_1| = |M^*_2| + (1/2) \cdot |M^*_1|.
		\end{align} Furthermore, since $M_1$ is an $\epsilon'$-AMM in $G$ for $\eps'=\eps/8$, we have that $|M^*_0| \leq \epsilon' \cdot \mu(G)$, since at most $\eps'\cdot \mu(G)$ nodes of $G$ must be removed from $G$ to make $M_1$ maximal, and at least one endpoint of each $M^*_0$ edge must be removed to achieve the same effect. 
		But then, since $M^*_0,M^*_1,M^*_2$ partition $M^*$, whose cardinality is $|M^*|=\mu(G)$, this implies that 
		\begin{align}\label{eqn:M^*_i-vs-OPT}
		|M^*_1| + |M^*_2| \geq (1-\epsilon') \cdot \mu(G).
		\end{align}
		Now, by \Cref{lm:maximal_matching_size}, 
		since $M^*_1$ is a matching in graph $G':=G[V(M_1),\overline{V(M_1)}]$ and $k\geq \frac{1}{\eps' b}$, we have
		\begin{align}\label{eqn:M2-vs-M^*_i}
		|M_2| \geq  \mu(G') \cdot \frac{k \cdot \lfloor k b \rfloor }{k+\lfloor k b\rfloor } \geq  \frac{\lfloor k b\rfloor }{b+1} \cdot |M^*_1| \geq  \frac{k b(1-\eps')}{b+1} \cdot |M^*_1|.
		\end{align}
		Combining equations \eqref{eqn:M1-vs-M^*_i}, \eqref{eqn:M^*_i-vs-OPT} and \eqref{eqn:M2-vs-M^*_i}, we obtain the following lower bound on our output estimate.
		\begin{align*}
		(1-\delta) \cdot |M_1|  +  (\delta/k) \cdot |M_2| & \stackrel{\eqref{eqn:M1-vs-M^*_i}, \eqref{eqn:M2-vs-M^*_i}}{\geq} (1 - \delta) \cdot \left(|M^*_2|  + (1/2) \cdot |M^*_1| \right) + (\delta/k) \cdot \frac{k b (1-\eps')}{b+1}  \cdot |M^*_1| \\
		& =  (1 - 1/b) \cdot |M^*_2| +  \left(\frac{1-1/b}{2} + \frac{1-\eps'}{b+1} \right) \cdot |M^*_1| \\
		& \geq (|M^*_1|+|M^*_2|) \cdot \min \left\{1 - 1/b,\, \frac{1-1/b}{2} + \frac{1-\eps'}{b+1} \right\}  \\
		& \stackrel{\eqref{eqn:M^*_i-vs-OPT}}{\geq} (1-\epsilon') \cdot \mu(G) \cdot \min \left\{1 - 1/b, \, \frac{1-1/b}{2} + \frac{1-\eps'}{b+1} \right\} \\
		& \geq (1-2\eps') \cdot \mu(G) \cdot \min \left\{1 - 1/b, \, \frac{1-1/b}{2} + \frac{1}{b+1} \right\} \\
		& = (1-2\eps')\cdot (2-\sqrt{2})\cdot \mu(G), 
		\end{align*}
		where the last equality follows by our choice of $b=1+\sqrt{2}$.  
		
		Thus, combining with \Cref{bip-bound}, and using that $\eps'=\eps/8 < 1/8$, we find that the output matching size estimate $\nu:=(1-\delta)\cdot |M_1|+(\delta/k)\cdot |M_2|$ is indeed a $\left(1+\frac{1}{\sqrt{2}}+\eps\right)$-approximation.
		\begin{align*}
		\nu & \leq \mu(G)\leq \nu \cdot \left(\frac{1}{(2-\sqrt{2})\cdot (1-2\eps')}\right)\leq \nu \cdot \left(\left(1+\frac{1}{\sqrt{2}}\right)\cdot (1+4\eps')\right) \leq \nu \cdot \left(1+\frac{1}{\sqrt{2}}+\eps\right). \qedhere
		\end{align*}
	\end{proof}
	
	\begin{remark}
		A direct extension of the tight example of \cite{konrad2021two} proves that this analysis is tight, up to the exact dependence on $\epsilon$.
	\end{remark}	
	
	\Cref{bip-bound} implies a 2-pass streaming algorithm for computing a $(1+\frac{1}{\sqrt{2}}+\eps)$-approximate maximum matching: simply store $G[M_1\cup M_2]$ and output a maximum matching in this subgraph by the stream's end. The space used in the first and second passes are $\tilde{O}(n)$ and $\tilde{O}(nkb)=\tilde{O}(n/\eps)$, respectively. 
	More interestingly for our goals, we show in the next section that \Cref{bip:approx} can be used to obtain a \emph{dynamic} approximation of the same quality, in polylogarithmic update time.

	\subsection{Dynamic Algorithm}
	
In this section, we show how to (approximately) implement \Cref{alg:bipartite} in polylogarithmic update time in a dynamic setting.
	
	\begin{thm}\label{bipartite-query-algorithm}
	    Let $\eps\in (0,1)$. There exists a robust dynamic algorithm $\calA$ with worst-case update time $t_u = \tilde{O}_{\eps}(1)$ w.h.p.~and query time $t_q=\tilde{O}_{\eps}(n)$ that outputs w.h.p.~a value $\nu\in [\mu(G)/ (1+\frac{1}{\sqrt{2}}+\eps), \mu(G)]$. That is, it answers $(1+\frac{1}{\sqrt{2}}+\eps)$ approximate matching size estimate queries.
	\end{thm}

\begin{proof} 
The dynamic algorithm $\calA$ is based on  \Cref{alg:bipartite}. Let $\eps'=\eps/12$. Throughout the updates, Algorithm $\cal{A}$ simply maintains an $(\eps'/8)$-AMM in $G$, denoted by $M_1$, invoking \Cref{eps-AMM-algo}. This immediately implies the desired update time of $t_u = \tilde{O}_{\epsilon}(1)$.

We now describe how Algorithm $\calA$ responds to a query about the maximum matching size. To answer this query, the algorithm considers a new auxiliary graph $G^* = (V^*, E^*)$, which is defined as follows. Set $b = 1+\sqrt{2}$. For each  $u \in V(M_1)$, create $k:=\lceil \frac{8}{\epsilon' b} \rceil$ copies of the node $u$ in $G^*$. Next, for each  $v \in \overline{V(M_1)}$, create $\lfloor k \cdot b \rfloor$ copies of the node $v$ in $G^*$.  Finally, for every edge $(u, v) \in G\big[ V(M_1), \overline{V(M_1)}\big]$, create an edge in $G^*$ between every pair $(u^*, v^*)$ of copies of the nodes $u, v$. Note that there is a one-to-one mapping between maximal matchings in the new graph $G^*$ and maximal $b$-matchings in $G\big[ V(M_1), \overline{V(M_1)} \big]$.

We emphasize that our dynamic algorithm $\calA$ does not explicitly maintain the  auxiliary graph $G^*$.  When we receive a query about the maximum matching size in $G$, we  explicitly construct only the node-set $V^*$ of $G^*$, based on the matching $M_1$. This takes only $O_{\epsilon}(n)$ time since $|V^*| = O_{\epsilon}(n)$. We can, however,  access the edges of $G^*$ by using adjacency matrix queries: there exists an edge $(u^*, v^*)$ in $G^*$ iff there exists an edge $(u, v)$ between the corresponding nodes in $G$.

At this point, we invoke \Cref{sublinear} with $G = G^*$ and precision parameter $e'' = (\epsilon')^3$. This gives us a value $\psi$, which is an estimate of $|M_2|$.  We now return $\nu := (1-1/b)\cdot |M_1| + (1/bk)\cdot \psi$ as our estimate of $\mu(G)$. Clearly, this entire procedure for answering a query takes $\tilde{O}_{\epsilon}(n)$ time. It now remains to analyze the approximation ratio. 		
Towards this end, we again appeal to \Cref{sublinear}. This proposition asserts that the value $\psi$ satisfies
		\begin{align}\label{eqn:psi}
			|M_2|\geq \psi \geq |M_2|-\epsilon'' \cdot |V^*|  \geq |M_2|-(\epsilon')^2 \cdot 16 n \geq |M_2| - \epsilon'\cdot \mu(G),
		\end{align}

		Therefore, our estimate $\nu$ satisfies that $\nu'\geq \nu\geq \nu'-\epsilon\cdot \mu(G)$, where
		\begin{align}\label{eqn:nu'}
			\nu' := (1-1/b)\cdot |M_1| + (1/bk)\cdot |M_2|.
		\end{align} 
		Now, by \Cref{bip:approx} and our choice of $k=\lceil \frac{8}{\epsilon' b}\rceil$ and $b=1+\sqrt{2}$, we have that 
		$\nu'\leq \mu(G)\leq \nu'\cdot \big(1+\frac{1}{\sqrt{2}}+\epsilon\big)$, from which we obtain that $\nu\leq \nu' \leq \mu(G)$ and moreover
		\begin{align*}
			\mu(G)\leq \nu'\cdot \left(1+\frac{1}{\sqrt{2}} +\epsilon'\right) \leq (\nu + \eps'\cdot \mu(G))\cdot \left(1+\frac{1}{\sqrt{2}} +\epsilon'\right) \leq \nu \cdot \left(1+\frac{1}{\sqrt{2}} +\epsilon'\right) + 3\eps'\cdot \mu(G).
		\end{align*}
	    Rearranging terms, and using that $\eps'=\eps/16\leq 1/12$, we have that \begin{align*}
	        \nu\leq \mu(G)\leq \nu\cdot \left(1+\frac{1}{\sqrt{2}}+\eps'\right)/(1-3\eps)\leq \nu\cdot \left(1+\frac{1}{\sqrt{2}}+\eps\right)\cdot (1+4\eps)\leq \nu\cdot \left(1+\frac{1}{\sqrt{2}}+16\eps'\right).
	    \end{align*}
	    That is, since $\eps'=\eps/16$, the estimate $\nu$ output after a query is  $\big(1+\frac{1}{\sqrt{2}}+\eps\big)$-approximate, w.h.p.
	\end{proof}
	
	Combining \Cref{bipartite-query-algorithm} and \Cref{reduction}, we obtain our result for bipartite graphs.
	
	\begin{thm}\label{bip-theorem}
		For any $\epsilon\in (0,1)$, there exists a $(1+\frac{1}{\sqrt{2}}+\eps)\approx (1.707+\eps)$-approximate randomized dynamic bipartite matching size algorithm with $\tilde{O}_{\eps}(1)$-update time.
		The algorithm's approximation ratio holds w.h.p.~against an adaptive adversary.
	\end{thm}
	
	\begin{remark}
	By the same approach as the recent dynamic weighted matching framework of \cite{bernstein2021framework} restricted to bipartite graphs, \Cref{bip-theorem} implies a $\big(1+\frac{1}{\sqrt{2}}+\eps)$ robust approximation for \emph{weighted} bipartite matching with the same update time, up to an exponential blowup in the dependence on $\eps$.\footnote{This extension is not obtained by using the framework of \cite{bernstein2021framework} directly, as the latter requires explicit dynamic matchings. Nonetheless, their arguments can be extended to the value version of the problem.}
	\end{remark}
	\section{Algorithms on General Graphs}
	\label{sec:algo-general}
	
	In this section we present our main result: a robust dynamic algorithm 
	maintaining a $(2-\Omega(1))$-approximation to the size of the maximum matching in a general graph in worst-case $\tilde{O}_{\epsilon}(1)$ update time. As with our algorithm for bipartite graphs, we start with a two-pass semi-streaming algorithm in \Cref{sec:general-streaming}, and then show how to approximately implement it dynamically in \Cref{subsec:dyn:general}. Finally, in  \Cref{sec:general-tradeoffs} we show that our approach allows us to improve the approximation of any algorithm with approximation ratio in the range $(1.5,2]$.
	
	\subsection{Two-Pass Streaming Algorithm}\label{sec:general-streaming}
	\label{sec:algo-general-streaming}
	
    The key challenge in extending \Cref{alg:bipartite} and its analysis to non-bipartite graphs is its reliance on the integrality of the fractional matching polytope in bipartite graphs. 
	This allowed us to focus on proving the existence of a large fractional matching, which guarantees the existence of a large integral matching of (at least) the same size.
	For general graphs this argument fails, and so instead we search for length-three augmenting paths ($3$-augmenting paths, for short) with respect to our first matching, $M_1$, by computing some large $b$-matching $M_2$ in the edge set $V(M_1)\times \overline{V(M_1)}$. 
	The main difficulty with this approach in general graphs is that both the endpoints of an edge $(u, v) \in M_1$  might  get matched (in $M_2$) to the same node $w$, and the resulting triangle $w - u - v - w$ does not help us in any way to create a $3$-augmenting path involving the edge $(u, v) \in M_1$.

	We overcome this difficulty   using \emph{random bipartitions} (see \Cref{alg:general-streaming}). As before, in the first pass we compute an $(\epsilon/4)$-AMM $M_1$ in the input graph $G$.\footnote{As with the bipartite \Cref{alg:bipartite}, we suggest the reader think of $M_1$ as being maximal for now (i.e., $\eps=0$).}
    Next, we define the following random bipartition $(L, R)$ of the node-set $V$. 
    For each matched edge $(u, v) \in M_1$, we arbitrarily include one of its endpoints in $L$ and the other in $R$. Next, for each unmatched node $v \in \overline{V(M_1)}$, we include the node $v$ in into one of $L$ and $R$ chosen uniformly at random. Subsequently, we assign a capacity $b_v := 1$ to all nodes $v \in V(M_1)$ and a capacity $b_v := b$ to all nodes $v \in \overline{V(M_1)}$, for some integer $b$ to be chosen later. Let $B=(V,E_2)$ be the  bipartite subgraph spanned by edges with a single node matched in $M_1$ and endpoints in opposite sides, i.e., 
    $$E_2:=\{(u,v)\in E \mid u\in V(M_1), \, v\in \overline{V(M_1)},\, |\{u,v\}\cap L|=|\{u,v\}\cap R|=1\}).$$
    In the second pass, we compute a maximal $b$-matching $M_2$ in $B$ w.r.t.~the capacities $\{ b_v \}$.  Finally, we return the maximum matching in the subgraph $G[M_1 \cup M_2]$.

	\begin{algorithm}[H]
		\caption{General Two-Pass Streaming Algorithm}
		\label{alg:general-streaming}
		\begin{algorithmic}[1]
			\State $M_1\gets$ $(\eps/4)$-AMM  computed from first pass
			\For{edge $(u,v)\in M_1$}
			\State  $s(u) \gets \ell$ and $s(v) \gets r$
			\EndFor
			\For{vertex $w\in \overline{V(M_1)}$}
			\State $s(w)\sim \Uni \{\ell,r\}$
			\EndFor
			\State let $L\gets \{v \mid s(v)=\ell\}$ and $R\gets \{v \mid s(v)=r\}$
			\State Assign each vertex $v$ capacity $b_v=\begin{cases}
			1 & v\in V(M_1) \\
			b & v\in \overline{V(M_1)}.
			\end{cases}$
			\State let $B\gets (V,E_2)$, for $E_2=\{(u,v)\in E \mid u\in V(M_1), v\in \overline{V(M_1)}, |\{u,v\}\cap L|=|\{u,v\}\cap R|=1\})$.
			\State $M_2\gets $ maximal $b$-matching in $B$ computed from second pass 
			\State Output maximum matching in $G[M_1\cup M_2]$
		\end{algorithmic}
	\end{algorithm}	
    
    \noindent\textbf{Intuition:} The intuition behind \Cref{alg:general-streaming} is as follows: if $M_1$ is only roughly $2$-approximate, then many $3$-augmenting paths exist in $G$ w.r.t.~$M_1$, by \Cref{poor-approx-maximal-matching=many-aug-paths}. Now, a constant fraction of these (specifically, a quarter) ``survive'' the random bipartition and have their extreme edges belong to $B$. Now, for each augmenting path $u'-u-v-v'$ that survives, either an augmenting path containing $u-v$ is found in $M_1\cup M_2$, or at least one of $u'$ or $v'$ is matched $b$ times to nodes in $V(M_1)$ other than $u$ or $v$. Next, since nodes in $V(M_1)$ can only be matched once in $M_2$, this limits the number of paths where $u$ and $v$ do not participate in an augmenting path. This implies a large number of augmenting paths in $M_1\cup M_2$ that are disjoint in their $V(M_1)$ nodes. 
    Finally, since each node in $\overline{V(M_1)}$ belongs to at most $b$ such paths, some large $\Omega(1/b)$ fraction of these augmenting paths are also disjoint in their $\overline{V(M_1)}$ nodes, from which we conclude that $M_1\cup M_2$ contains a large set of node-disjoint augmenting paths w.r.t.~$M_1$, and that $G[M_1\cup M_2]$ contains a large matching. 
    \medskip 
    
    We now substantiate the above intuition. The first lemma in this vein asserts that $M_1\cup M_2$ contains many $3$-augmenting paths w.r.t.~$M_1$ (assuming $M_1$ is not already near maximum in size).
    
    \begin{restatable}{lem}{augtighter}\label{general:many-aug-paths-tighter}
		If $|M_1|=\left(\frac{1}{2}+c\right)\cdot \mu(G)$, then 
		$G[M_1\cup M_2]$ contains a set $\calP$ of $3$-augmenting paths w.r.t.~$M_1$ that are disjoint in their $V(M_1)$ nodes, with expected cardinality at least 
		$$\E[|\calP|]\geq \left(\frac{b}{b+1}\right)\cdot \left( \frac{1}{4}\cdot \left(\frac{1}{2}-3c\right) -  \frac{1}{b} \cdot \left(\frac{1}{2}+c\right) -\frac{7\eps}{8} \right)\cdot \mu(G).$$
	\end{restatable}
	
	As the proof of \Cref{general:many-aug-paths-tighter} is a little calculation heavy, we defer its proof to \Cref{appendix:general}, and instead prove the following slightly weaker but simpler bound here.
 
	\begin{lem}\label{general:many-aug-paths-simple}
		If $|M_1|=\left(\frac{1}{2}+c\right)\cdot \mu(G)$, then 
		$G[M_1\cup M_2]$ contains a set $\calP$ of $3$-augmenting paths w.r.t.~$M_1$ that are disjoint in their $V(M_1)$ nodes, with expected cardinality at least 
		$$\E[|\calP|]\geq \left(\frac{1}{4}\cdot\left(\frac{1}{2}-3c-\frac{7\eps}{2}\right) - \frac{2}{b}\cdot\left(\frac{1}{2}+c\right)\right)\cdot \mu(G).$$
	\end{lem}
	\begin{proof}
		Fix a maximum set of node-disjoint $3$-augmenting paths in $G$ w.r.t.~$M_1$, denoted by $\calP^*$. 
		By \Cref{poor-approx-maximal-matching=many-aug-paths}, we have $|\calP^*|\geq \left(\frac{1}{2}-3c-\frac{7\eps}{2}\right)\cdot \mu(G)$. 
		Next, let $S\subseteq \calP^*$ be the paths $u'-u-v-v'$ who ``survive'' the bipartition, in that $(u,u'),(v,v')\in E_2$.
		By construction, each path in $\calP^*$ survives with probability exactly $\frac{1}{4}$. Therefore, $\E[|S|]\geq \frac{1}{4}\cdot \left(\frac{1}{2}-3c-\frac{7\eps}{2}\right)\cdot \mu(G)$. 
		
		Now, for each survived path $u'-u-v-v'\in S$, either both $u$ and $v$ are matched (exactly once) in $M_2$, thus contributing an augmenting path, or at least one of $u'$ and $v'$ must be matched in $M_2$ to $b$ distinct nodes in $V(M_1)$. 
		But since each vertex in $V(M_1)$ is matched at most once in $M_2$, there are thus at most $|V(M_1)|/b=2|M_1|/b$ paths in $S$ whose middle edges do not belong to a $3$-augmenting path in $M_1\cup M_2$.
		Therefore, there are at least
		$|S| - 2|M_1|/b$ many edges $(u,v)$ in $M_1$ whose endpoints are both matched in $M_2$ to some (different) nodes $u''$ and $v''$, respectively. Each such edge contributes an augmenting path to a set $\calP$ of the desired size,
		\begin{align*}
		\E[|\calP|] = \E[|S|] - \frac{2|M_1|}{b} & \geq  \left(\frac{1}{4}\cdot\left(\frac{1}{2}-3c-\frac{7\eps}{2}\right) - \frac{2}{b}\cdot\left(\frac{1}{2}+c\right)\right)\cdot \mu(G). \qedhere 
		\end{align*}
	\end{proof}
	
	The preceding two lemmas imply the existence of a multitude of $3$-augmenting paths that are disjoint in their $V(M_1)$ nodes. 
	We now use these augmenting paths to prove the existence of numerous (though possibly fewer) augmenting paths that are disjoint in \emph{all} their nodes.
	Since each of the two $\overline{V(M_1)}$ nodes of a $3$-augmenting path belong to at most $b$ such paths, it is easy to find some $1/(2b-1)$ fraction of these augmenting paths that are disjoint in all their nodes.
	The following lemma, resembling \cite[Lemma 6]{esfandiari2016finding}, increases this fraction to $1/b$.

	\begin{lem}\label{many-disjoint-tighter}
	    Let $\calP$ be a set of $3$-augmenting paths w.r.t.~$M_1$ in $G[M_1 \cup M_2]$ such that each $V(M_1)$ (resp. $\overline{V(M_1)}$) node belongs to at most one (resp., $b$) paths in $\calP$. Then $\calP$ contains a set of node-disjoint $3$-augmenting paths $\calP'\subseteq \calP$ of cardinality at least $|\calP'|\geq \frac{1}{b}\cdot |\calP|$.
	\end{lem}
	\begin{proof}
	    Consider the graph $G'=(\overline{V(M_1)},E')$ obtained by replacing each path $u'-u-v-v'$ in $\calP$ with a single edge $u'-v'$. This graph $G'$ is bipartite, by virtue of our random bipartition of $G$. Now, since this bipartite graph $G'$ has maximum degree $b$, it contains a matching of size at least $|E'|/b=|\calP|/b$: the fractional matching assigning values $1/b$ to each edge has value $|E'|/b$, and so $G'$ contains an \emph{integral} matching of at least the same value. On the other hand, disjoint edges in $G'$ have a one-to-one mapping to node-disjoint paths in $\calP$, since each node in $V(M_1)$ belongs to at most one such path. Thus, the maximum matching in $G'$ corresponds to a collection $\calP'\subseteq\calP$ of node-disjoint augmenting paths in $G[M_{1}\cup M_2]$ w.r.t.~$M_{1}$, of cardinality at least $|\calP'|\geq|\calP|/b$. 
    \end{proof}
	
	The three preceding lemmas imply that $G[M_1\cup M_2]$ contains a large set of vertex-disjoint $3$-augmenting paths w.r.t.~$M_1$, assuming this latter matching is not already large. As we now show, this implies that $G[M_1\cup M_2]$ contains a better-than-2-approximate matching.
	
	\begin{thm}\label{streaming-general-graphs-improed}
		Let $\eps\in (0,1/4)$. Then, \Cref{alg:general-streaming} with $b=9$ satisfies $\mu(G)\geq \E[\mu(G[M_1\cup M_2])]\geq \left(\frac{1}{2}+\frac{1}{144}-\eps \right)\cdot \mu(G)$, and is thus $(\frac{1}{2}+\frac{1}{144} - \eps)^{-1} < 1.973(1+2\eps)$-approximate in expectation.
	\end{thm}
	\begin{proof}
		Let $|M_1|=\left(\frac{1}{2}+c\right)\cdot \mu(G)$, where $c\in[-\eps/8, 1/2]$, with the lower bound on $c$ following from \Cref{AMM-size-bound} and $M$ being an $(\eps/4)$-AMM.
		Let $\calP'$ be a maximum set of vertex-disjoint $3$-augmenting paths w.r.t.~$M_1$ in $G[M_1\cup M_2]$.
		Then, augmenting along these paths, we find that $M_1\bigoplus \calP'$ contains a matching (hence of size at most $\mu(G)$) of the desired expected cardinality.
		\begin{align*}
		\E[|M_1 \bigoplus \calP'|] & =\E[ |M_1|+ |\calP'|] \\
		& \geq \left(\frac{1}{2}+c\right)\cdot \mu(G) + \left(\frac{1}{b+1}\cdot \left(\frac{1}{4}\cdot\left(\frac{1}{2}-3c\right) - \frac{1}{b}\cdot\left(\frac{1}{2}+c\right)-\frac{7\eps}{8}\right)\right)\cdot \mu(G) \\
		& \geq 
		\left(\frac{1}{2}-\frac{\eps}{8}+\frac{1}{b+1}\cdot\left( \frac{1}{4}\cdot\frac{1}{2} - \frac{1}{b} \cdot \frac{1}{2}\right) -\frac{7\eps}{8} \right)\cdot \mu(G) \\
		& = \left(\frac{1}{2}+\frac{1}{144} -\eps \right)\cdot \mu(G).
		\end{align*}
		Above, the first inequality follows from  \Cref{general:many-aug-paths-tighter} and \Cref{many-disjoint-tighter}, the second inequality mainly relies on the parenthetical expression being increasing in $c\geq -\eps/8$ (for our choice of $b=9$). Finally, the equality holds by our choice of $b=9$.
	\end{proof}

\subsection{Dynamic Algorithms}
\label{sec:general-dynamic}
\label{sec:dyn_general}

In this section we provide the dynamic algorithms yielding our main results, theorems \ref{main-theorem} and~\ref{thm:tradeoffs}. As with the bipartite case, our general approach is to approximately implement our two-pass streaming algorithm in a dynamic setting. Unlike the  algorithm for bipartite graphs, here we need to (slightly) unbox the sublinear-time algorithm of \cite{behnezhad2022time}  to find a large set of edges in $M_1$ which belong to $3$-augmenting paths in $G[M_1 \cup M_2]$, as explained below.

\subsubsection{Breaking the barrier of two in polylog time}
\label{subsec:dyn:general}

In this section, we present a robust dynamic $(1.973+\epsilon)$-approximate maximum matching size with worst case update time of $t_u=\tilde{O}_{\epsilon}(1)$, and a query time of $t_q=\tilde{O}_{\epsilon}(n)$, provided $\mu(G) \geq \epsilon n$. This, combined with \Cref{reduction}, implies our main result, \Cref{main-theorem}.

For our dynamic (approximate) implementation of \Cref{alg:bipartite}, which works on bipartite graphs, all we needed was to estimate $|M_2|$.
In contrast, for our dynamic (approximate) implementation of \Cref{alg:general-streaming}, we will need to estimate the size of the set $\calP$ as guaranteed by \Cref{general:many-aug-paths-tighter}. Specifically, we note that the proofs of lemmas \ref{general:many-aug-paths-tighter}, and \ref{many-disjoint-tighter} and \Cref{streaming-general-graphs-improed} imply the following observation.
\begin{obs}
\label{obs:estimate:general}
Let $\widehat{M}_1 \subseteq M_1$ be the set of edges in $M_1$ whose two endpoints are matched in $M_2$ in \Cref{alg:general-streaming} run with $b = 9$. Then, $|M_1| + \frac{1}{b} \cdot \E\left[ |\widehat{M}_1| \right] \geq  \frac{\mu(G)}{1.973 \cdot (1+2\epsilon)}$, and also $\mu(G) \geq |M_1| + \frac{1}{b} \cdot |\widehat{M}_1|$.
\end{obs}

To estimate $|\widehat{M}_1|$ efficiently, we make use of the following extension of the algorithm of \cite{behnezhad2022time}.

\begin{restatable}{lem}{lmsublinearextension}\label{lm:sublinear:extension}
Consider a graph $G' = (V', E')$ with $|V'| = n'$, and a matching $M$ with $V(M) \subseteq V'$ that is not necessarily part of $G'$ (i.e., we might have $M \nsubseteq E'$). 
For any matching $M'$ in $G'$, let $k_{M'}$ denote the number of edges in $M$ both of whose endpoints are matched in $M'$. There is an algorithm which, given adjacency matrix query access to the edges of $G'$,  whp  runs in $\tilde{O}_{\epsilon}(n')$ time and returns an estimate $\kappa \in  [k_{M'} - \epsilon^2 n',  k_{M'}]$ for some maximal matching $M'$ in $G'$.
\end{restatable} 

This lemma follows from the work of \cite{behnezhad2022time} rather directly, though it requires some unboxing of the results there, due to the organization of that work. 
We substantiate this lemma in \Cref{appendix:sublinear:extension}.

Given the above, we are now ready to prove the main result of this section, which is summarized in the theorem below.

\begin{thm}
\label{thm:main:dynamic}
    For any $\eps\in (0,1/4)$, there exists a robust dynamic matching size estimator algorithm $\calA$ with worst-case update time $\tilde{O}_{\eps}(1)$ that, provided $\mu(G)\geq \eps\cdot n$, supports $\tilde{O}_{\eps}(n)$-time queries and outputs a $(1.973+\eps)$-approximate estimate w.h.p.
\end{thm}

\begin{proof}
The dynamic algorithm $\calA$ is based on  \Cref{alg:general-streaming}. For its updates, it maintains an $(\epsilon/4)$-AMM $M_1$ in the input graph $G$, using \Cref{eps-AMM-algo}, and a balanced binary search tree (BST) of edges in the graph, allowing for logarithmic-time insertion, deletion and edge queries. 
This immediately implies a worst-case update time of $t_u=\tilde{O}_{\epsilon}(1)$.

We now describe how Algorithm $\calA$ responds to a query about the maximum matching size. To answer this query, the algorithm considers a new auxiliary graph $G^* = (V^*, E^*)$, which is defined as follows. For each node $u \in V(M_1)$, create a node $0_u$ in $G^*$. Next, for each node $v \in \overline{V(M_1)}$, create $b$ nodes $1_v, \ldots, b_v$ in $G^*$. Finally, for every edge $(u, v) \in E_2$, with $u \in V(M_1)$ and $v \in \overline{V(M_1)}$, create an edge $(0_u, i_v)$ in $G^*$ for all $i \in \{1, \ldots, b\}$. Note that there is a one-to-one mapping between maximal matchings in the new graph $G^*$ and maximal $b$-matchings in $B$.

We emphasize that our dynamic algorithm $\calA$ does not explicitly maintain the auxiliary graph $G^*$.  When we receive a query about the maximum matching size in $G$, we explicitly construct only the node-set $V^*$ of $G^*$, based on the matching $M_1$. This takes only $O(n)$ time. We can, however, simulate adjacency matrix queries in $G^*$ efficiently: there exists an edge $(0_u, i_v)$ in $G^*$ iff there exists an edge $(u, v)$ between the corresponding nodes in $G$, verifiable in $O(\log n)$ time using our edge-set BST.

At this point, we estimate the size of $\widehat{M}_1$ by invoking \Cref{lm:sublinear:extension} with $G' = G^*$ and $M = M_1$. This gives us, in time $\tilde{O}_{\epsilon}(n)$ a value $\kappa$ satisfying $\kappa\in [|\widehat{M}_1|-\eps^2n,|\widehat{M}_1|]$, w.h.p. We now return $\nu := |M_1|+\frac{1}{b}\cdot \kappa$ as our estimate of $\mu(G)$. All in all, our algorithm has query time $t_q = \tilde{O}_{\eps}(n)$. 

It remains to analyze the approximation ratio. Towards this end, we observe that, by our hypothesis that $\mu(G)\geq \eps\cdot n$,
\begin{align}
  \E[\nu] & \geq \E\left[|M_1|+\frac{1}{b}\cdot (|\widehat{M}_1| - \eps^2 n)\right] \nonumber \\
  & = |M_1| + \frac{1}{b} \cdot \E\left[ \widehat{M}_1 \right] - \frac{\epsilon^2 n}{b} \nonumber \\
  & \geq \frac{\mu(G)}{1.973 \cdot (1+2\epsilon)} - \epsilon^2 n \nonumber  \\
  & \geq \frac{\mu(G)}{1.973 \cdot (1+2\epsilon)} - \epsilon \cdot \mu(G) \nonumber \\
  & \geq \frac{\mu(G)}{1.973 \cdot (1+5 \epsilon)}. \label{eq:approx:general:1}
\end{align}
In the above derivation, the second inequality follows from~\Cref{obs:estimate:general}. Similarly, we  have whp:
\begin{align}
\nu & \leq |M_1|+\frac{1}{b}\cdot |\widehat{M}_1|  \leq \mu(G). \label{eq:approx:general:2}
\end{align}
The second inequality in the above derivation again follows from~\Cref{obs:estimate:general}. From~(\ref{eq:approx:general:1}) and~(\ref{eq:approx:general:2}), we conclude that we return in response to each query a $1.973 (1+5\epsilon)$-approximation to $\mu(G)$ in expectation.
Therefore, by standard Chernoff bounds, running $O_{\eps}(\log n)$ copies of this algorithm (increasing update and query time appropriately) and taking the average of these will then result in a $1.973(1+O(\eps))$ approximation of the desired value, w.h.p.
Reparameterizing $\eps$ appropriately, the theorem follows.
\end{proof}

Combined with \Cref{reduction}, the above theorem implies our main result, \Cref{thm:main:dynamic}.

\subsubsection{New time/approximation tradeoffs}\label{sec:general-tradeoffs}

	In this section we show our secondary
result: a black-box method to improve dynamic matching algorithm's
approximation ratio, at the cost of only outputting a size estimate.
We start with the following observation.
\begin{prop}
\label{maximal-and-approx} Let $G$ be an $n$-node graph, $\eps\in(0,1)$
and $\alpha\geq1$. Then, given an $\eps$-AMM $M'$ and $\alpha$-approximate maximum 
matching $M''$ in $G$, one can compute in $O(n)$ time a matching $M$ in $G$ which
is both $\alpha$-approximate and an $\eps$-AMM. 
\end{prop}

\begin{proof}
The subgraph $G[M'\cup M'']$ has maximum degree two, and is thus
the union of paths and cycles. Let $M$ be the matching obtained by
taking from each connected component $\calC$ in $G[M'\cup M'']$
either the set of edges of $M'$ or $M''$ that are most plentiful
in $\calC$, breaking ties in favor of $M'$. By construction, it
is clear that $M$ is a matching, and that moreover $|M|\geq|M''|$,
and so $M$ is $\alpha$-approximate. On the other hand, $M$ matches
all nodes of $M'$ in each component, and therefore overall. That
is, after removing at most $\eps\mu(G)$ nodes in $V\setminus V(M)\subseteq V\setminus V(M')$,
we obtain a graph in which $M$ is maximal. That is, the matching
$M$ is also an $\eps$-AMM. 
\end{proof}
We are now ready to prove \Cref{thm:tradeoffs}, restated below
for ease of reference.

\tradeoffs* 
\begin{proof}
Let $\eps>0$ be some sufficiently small constant. 
We describe how to obtain a new dynamic matching size estimator $\calA$ for $G$, with update time $\tilde{O}_{\epsilon}(t_u)$ that, provided $\mu(G)\geq \eps\cdot n$, supports $\tilde{O}_{\epsilon}(n)$-time queries and outputs a $\beta$-approximate estimate w.h.p., for some $\beta=\alpha-\Omega((1-6(1/\alpha-1/2))^2))$. 
The theorem then immediately follows from Proposition~\ref{reduction}.

The algorithm $\calA$ works as follows. 
It maintains an $\eps$-AMM $M'_1$, invoking Lemma~\ref{eps-AMM-algo}, taking $\tilde{O}_{\eps}(1)$ update time.
It also maintains $\alpha$-approximate matching $M''_1$, by running the dynamic algorithm guaranteed by the theorem's hypothesis, taking $t_u$ update time.
Therefore, Algorithm $\calA$ has an overall update time of $t_u + \Otil_{\eps}(1) = \Otil_{\epsilon}(t_u)$.

Upon receiving a query, Algorithm $\calA$ first invokes Proposition~\ref{maximal-and-approx} to obtain a matching $M_1$ (based on $M_1'$ and $M_1''$) that is simultaneously an $\epsilon$-AMM and an $\alpha$-approximate maximum matching in $G$. This takes $O(n)$ time. The rest of the query algorithm remains exactly the same as in \Cref{sec:dyn_general}. This implies that Algorithm $\calA$ has an overall query time of $\tilde{O}_{\epsilon}(n)$. 

We now analyze the approximation guarantee of $\calA$.
Towards this end, observe that as $\alpha>1.5$, we can write $\frac{1}{\alpha}=\frac{1}{2}+c$ where
$0 < c<1/6$. 
So, we have that $|M_{1}| = \left(\frac{1}{2}+z\right)\cdot \mu(G)$ for some $z\geq c$.
Therefore, by Lemma~\ref{general:many-aug-paths-simple} and Lemma~\ref{general:many-aug-paths-simple}, there exists some set $\calP'$ of node-disjoint length-three augmenting paths w.r.t.~$M_1$ in $G[M_1\cup M_2]$ whose cardinality satisfies
$$\frac{|\calP'|}{\mu(G)} \geq f_b(z) :=\frac{1}{b}\cdot\left(\frac{1}{4}\cdot\left(\frac{1}{2}-3z-\frac{7\eps}{2}\right) - \frac{2}{b}\cdot\left(\frac{1}{2}+z\right)\right).$$
Augmenting along these paths with respect to $M_1$, we obtain a new matching in $G[M_1\cup M_2]$ of cardinality at least $\left(\frac{1}{2}+z+f_b(z)\right)\cdot \mu(G)$.
Now, for $b\geq 2$ (as we will choose), this matching size is decreasing in $z$, as observed by taking the derivative of $1/2+z+f_b(z)$ w.r.t.~$z$.
Therefore, the matching size is minimized at $z=c$, and we find that $\mu(G[M_1\cup M_2])\geq 1/2+c+f_b(c)$. 
Taking another derivative, this time with respect to $b$, we find that this expression is minimized  (ignoring the $\eps$ dependence) at $b^*=\frac{16 (1 + 2c)}{1 - 6c}$. Note that $b\geq 16$, as $c\in (0,1/6)$. This optimal $b^*$ need not be an integer, however, and so we take $b=\lceil b^*\rceil\leq \frac{17}{16}b^*$ in our algorithm, and find that $M_1\cup M_2$ contains a matching of size at least $\mu(G)$ times $1/2+z+f_{b}(z)\geq 1/2+c+f_{\frac{17b}{16}b^*}(c)\geq 1/2+c+\frac{9 (1 - 6 c)^2}{2312 (1 + 2 c))} - O(\eps/b)$. 
Moreover, some $b\cdot f_b(z)$ many edges $\hat{M}_1\subseteq M_1$ have both of their endpoints matched in the $b$-matching $M_2$.

We conclude that $\E[\mu(G[M_1 \cup M_2])]$ gives a strictly-better-than-$\alpha$ approximation to $\mu(G)$ (again using that $c<1/6$). 
Specifically, the gain we get in the approximation ratio is of the order of $\Theta((1-6c)^2)=\Theta((1-6(1/\alpha-1/2))^2)$.
Now, using the fact that $\mu(G)\ge\eps n$ and we are running the same query algorithm
as in \Cref{sec:dyn_general}, our estimation using the sub-linear-time algorithm (Lemma~\ref{lm:sublinear:extension})
gives a strictly-better-than-$\alpha$ approximation to $\mu(G)$ in
expectation. 
As before, taking the average of $O(\log n)$ copies of this algorithm will provide the same bound w.h.p., at an additional logarithmic multiplicative overhead to the update and query times.
\end{proof}

	\begin{remark}
	    We note that the reduction underlying \Cref{thm:tradeoffs} preserves robustness and worst-case update time.
	\end{remark}
	
	\section{AMMs against Adaptive Adversaries}\label{sec:adaptive}
	In this section we prove \Cref{eps-AMM-algo}. That is, we provide a robust dynamic algorithm for maintaining an $\eps$-AMM in worst-case polylogarithmic update time.  But first, we motivate our algorithm, and characterize the kind of matching we wish to compute. 
	
	We first recall a useful tool in the literature, namely \emph{edge} sparsification: maintaining a sparse subgraph of $G$ containing a large matching--a so-called \emph{matching sparsifier}. Such sparsifiers naturally allow to achieve speedups in the algorithms needed for \Cref{reduction}, as a large matching in a sparsifier can be computed quickly. 
	One influential such sparsifier that we will use are \emph{kernels}, introduced by Bhattacharya et al.~\cite{bhattacharya2018deterministic}.
	\begin{Def}\label{def:kernel}
		For $\eps\geq [0,1]$ and $d\in \mathbb{N}$, a subgraph $K=(V,E_K)$ of graph $G=(V,E)$ is an \emph{$(\eps,d)$-kernel} if $K$'s maximum degree is at most $d$ and each edge $e\in E\setminus E_K$ has at least one endpoint of degree at least $d(1-\eps)$ in $K$. 
	\end{Def}
	
	These sparsifiers will play an integral role in robustly and efficiently maintaining an AMM in this section. We start by motivating their use in computing AMMs in a \emph{static} setting.
	
	\subsection{From kernels to AMMs}
	To motivate the interest in bounded-degree graphs, we recall the following observation, which follows from the $2\mu(G)$ endpoints of a maximum matching forming a vertex cover (i.e., being incident on each edge of the graph.
	
	\begin{fact}\label{size-degree}
	Let $G=(V,E)$ be a graph of maximum degree $\Delta$. Then $|E|\leq 2\mu(G)\cdot \Delta.$
	\end{fact}
	Consequently, for small $d$ we have that $(\eps,d)$-kernels of $G$ are sparse subgraphs. In particular, the time needed to compute maximal matchings in such subgraphs is linear in their size, $|E_K|\leq 2\mu(K)\cdot d = O(\mu(G)\cdot d)$.
	The following result of \cite{duan2014linear} implies that essentially the same amount of time is needed to compute a \emph{near-maximum-weight} matching in $K$.
	\begin{prop}\label{linear-mwm}
	    Let $G=(V,E,w)$ be a weighted graph. Then, one can compute deterministically a $(1+\eps)$-approximate maximum weight matching in $G$ in time $O_{\eps}(|E|)$.
	\end{prop}
	
	We now turn to identifying useful matchings in a kernel $K$ that allow us to obtain an AMM of $G$. For this, we will need to upper bound the number of high-degree nodes in $K$.
	Specifically, for an $(\eps,d)$-kernel $K$ of graph $G$, we denote by $H_K:=\{v \mid d_K(v)\geq d(1-\eps)\}$ the set of high-degree nodes in $K$.
    We will wish to argue that a removal of few high-degree nodes in the kernel yields a subgraph in which our matching is maximal. We therefore need to prove that the number of high-degree nodes is itself small in terms of $\mu(G)$.
	
	\begin{lem}\label{few-high}
		Let $K=(V,E_K)$ be an $(\eps,d)$-kernel $K$ of $G$ with $\eps\leq 1/4$. Then $|H_K|\leq 4\mu(G)$.
	\end{lem}
	\begin{proof}
		We consider the fractional matching $x \in \mathbb{R}^E$ where $x_e = \mathds{1}[e\in E_K]/d$.
		By the degree bound of $K$, this is a feasible fractional matching in $G$.
		Using this fractional matching, we can show that
		$$\frac{1}{2}\cdot (1-\eps)\cdot |H_K| \leq
		\sum_{v\in H_K} \sum_{e\ni v} x_e \leq \sum_e x_e \leq \frac{3}{2}\cdot \mu(G),$$
		where the first inequality follows from the definition of $H_K$ and possible double counting of edges, and the last inequality follows from the integrality gap of $\frac{3}{2}$ of the factional matching polytope.
		Simplifying the above and using $\eps\leq 1/4$, we have that indeed
		$|H_K|\leq (3/(1-\eps))\cdot \mu(G)\leq 4\mu(G)$.
	\end{proof}
	
	We now characterize the matchings that we wish to compute in this section, prove that they exist and that they are AMMs.
	\begin{lem}\label{the-adaptive-AMM}
		Let $K=(V,E_K)$ be an $(\eps,d)$-kernel of $G=(V,E)$, for $\eps\in (0,1/4)$ and $d\geq \frac{1}{\eps}$. Then, a maximal matching $M$ in $K$ that matches at least a $(1-c\cdot \eps)$-fraction of $H_K$ is a $4c\eps$-AMM in $G$. Moreover, such a matching exists for $c=2$.
	\end{lem}
	\begin{proof}
		First, we argue that such a matching $M$, if it exists, is indeed a $4c\eps$-AMM in $G$.
		We recall that every edge in $E\setminus E_K$ has a high-degree endpoint in $K$.
		Therefore, if we remove the $c\eps$ fraction of high-degree nodes $H_K$ unmatched by $M$, each edge in $E\setminus E_K$ in the resulting graph $G'$ has at least one endpoint matched in $M$. On the other hand, every edge in $E_K$ has an endpoint matched in $M$, by maximality of $M$ in $K$.
		We conclude that after removing $c\eps\cdot |H_K|\leq 4c\eps\cdot \mu(G)$ nodes in $G$ (with the inequality relying on \Cref{few-high}), we obtain a graph $G'$ where $M$ is maximal.
		That is, $M$ is a $4c\eps$-AMM.
		
		We now argue the existence of such a matching $M$ for $c=2$.
		Since $H$ has maximum degree $d\geq \frac{1}{\eps}$, by Vizing's theorem \cite{vizing1964estimate} it can be $(d+1)$-edge-colored, i.e., decomposed into $(d+1)$ matchings. A randomly-chosen color in this edge coloring is a matching $M'$ that matches each edge with probability $\frac{1}{d+1}$, and thus it matches each high-degree vertex $v$ with probability at least 
		$$\Pr[v \textrm{ matched}]\geq (1-\eps)/(d+1)\geq (1-\eps)/(1+\eps) \geq 1-2\eps.$$
		Finally, extending this matching $M'$ to also be maximal in $K$ by adding edges of $K$ greedily then proves the existence of the desired $8\eps$-AMM contained in the kernel $K$.
	\end{proof}
	
	The above implies a static algorithm with running time $\tilde{O}_{\eps}(d\cdot \mu(G))$ for computing an $\eps$-AMM in $G$ given an $(\eps,d)$-kernel $K$ of $G$.
	
	\begin{lem}\label{static-AMM}
	    Given an $(\eps,d)$-kernel $K=(V,E_K)$ of $G=(V,E)$, one can compute an $\eps$-AMM in $G$ in deterministic time $O_{\eps}(d\cdot \mu(G))$.
	\end{lem}
	\begin{proof}
	By \Cref{size-degree}, the number of edges in $K$ is at most $|E_K|=O(d\cdot \mu(G))$. We then compute a $(1+\eps')$-max weight matching $M'$ in the graph $G$ with edge weights equaling the number of high-degree nodes incident on them, $w_e = \sum_{v\in e} \mathds{1}[v\in H_K] \in \{0,1,2\}$.
	By \Cref{linear-mwm}, this can be done in deterministic time $O_{\eps}(d\cdot \mu(G))$.
	By \Cref{the-adaptive-AMM}, this guarantees that at least a $(1-2\eps')/(1+\eps')\geq (1-3\eps')$ fraction of high-degree nodes in $K$ are unmatched by this dynamic subroutine.
	We then extend $M'$ to also be maximal in $K$, by scanning over the $|E_K|=O(d\cdot \mu(G))$ edges of $K$ (in the same deterministic time) and adding them to $M'$ where possible. 
	By \Cref{the-adaptive-AMM}, this results in a $12\eps'$-AMM, i.e., an $\eps$-AMM, after a total of $O_\eps(d\cdot \mu(G))$ deterministic time.
	\end{proof}
	
	So far, we have provided a \emph{static} AMM algorithm with deterministic time $O_{\eps}(d\cdot \mu(G))$, \emph{provided we have access to a kernel}.
	To dynamize the above, we fisrt show how to (periodically) compute a kernel dynamically.
	
	\subsection{Periodic kernels and AMMs}
	In \cite{wajc2020rounding}, Wajc provided a method for rounding dynamic fractional matchings to matching sparsifiers, and from these (by methods underlying \Cref{alg:reduction}), we can obtain  integral matchings. Crucially for our needs, his framework was robust, and allowed for worst-case update times.
	Unfortunately for us, the lemma statements in his work do not immediately imply a robust dynamic kernel maintenance.
	However, they do allow for \emph{kernel queries}, with running time $\tilde{O}(d\cdot \mu(G))$.

	\begin{restatable}{lem}{periodickernel}\label{periodic-kernel}
	    Let $\eps\in (0,1)$ and $d=\tilde{O}_{\eps}(1)$ be sufficiently large. Then, there exists a robust algorithm with worst-case update time $t_u = \tilde{O}_{\eps}(1)$ allowing for $(\eps,d)$-kernel and $\eps$-AMM queries in worst-case query time $t_q=\tilde{O}_{\eps}(d\cdot \mu(G))$. The query's outputs are a kernel and an $\eps$-AMM w.h.p.
	\end{restatable}
    Given the ability to query a kernel, the ability to query an AMM then follows directly from \Cref{static-AMM}.
    As the proof and presentation of an algorithm allowing for kernel queries essentially requires repeating verbatim numerous lemmas in \cite{wajc2020rounding}, we defer its proof to \Cref{appendix:adaptive}. 
    
    We now turn to designing a robust dynamic algorithm that \emph{always} maintains an AMM.
	
	\subsection{Robust dynamic AMMs}
	
	So far, we have provided a method to answer AMM queries in a dynamic setting. 
	To lift this result to obtain AMM maintenance algorithms, we can rely on 
	
	\begin{lem}\label{AMM-stable}
	Let $\eps\in (0,1/2)$. If $M$ is an $\eps$-AMM in $G$, then the non-deleted edges of $M$ during any sequence of at most $\eps \cdot \mu(G)$ updates constitute a $6\eps$-AMM in $G$ (during the updates).
	\end{lem}
	\begin{proof}
	    Let $G$ and $M$ be the graph and matching before the updates, and let $G'$ and $M'$ be their counterparts after these updates. Since each update can decrease the size of the maximum matching size by at most one, we have that 
	    $$\frac{1}{2}\cdot \mu(G) \leq (1-\eps)\cdot \mu(G) \leq \mu(G').$$
	    Now, recall that for some set of vertices $U\subseteq V$ of size at most $|U|\leq \eps\cdot \mu(G)$ nodes from $G$, the matching $M$ is maximal in $G[V\setminus U]$.
	    Now, after these $\eps\cdot \mu(G)$ updates, it might be that $2\eps\cdot \mu(G)$ edges in $G'$ are now not incident on edges in $M'$. (The factor of two arises due to edges of $M$ that are deleted leaving two uncovered edges, addressable by removing two more nodes).
	    That is, after removing a node set $U'\subseteq V$ of size at most $|U'|\leq 3\eps\cdot \mu(G)\leq 6\eps\cdot \mu(G')$ nodes from $G'$, we obtain a graph $G'[V\setminus U']$ where $M'$ is maximal.
	    That is, $M'$ is an $O(\eps)$-AMM in $G'$.
	\end{proof}
	
	The above ``stability'' property of AMMs again lends itself to the periodic re-computation framework of \cite{gupta2013fully}, which, together with our algorithms for querying for AMMs, yields algorithms for maintaining AMMs (always).
	
	\begin{lem}
	    Let $\eps\in (0,1/2)$. Then, there exists a robust dynamic algorithm for maintaining an $\eps$-AMM w.h.p.~(at all times) in w.c.~update time $\tilde{O}_{\eps}(1)$.
	\end{lem}
	\begin{proof}
	    We will run the dynamic AMM query algorithm $\calA$ of \Cref{periodic-kernel}, whose update fits within our update time budget.
	    We will periodically query $\calA$, and spread this computation over these periods to guarantee low worst-case update time.
	    Specifically, we will divide the update sequence into \emph{epochs}, where if the graph $G$ at the start of epoch $i$ is $G_i$, then the epoch has length $\ell_i \in [\eps\cdot \mu(G)/3, \eps\cdot \mu(G_i)]$. 
	    In order to determine the length of the epochs, we run the deterministic dynamic $(2+\eps)$-approximate fractional matching algorithm of \cite{bhattacharya2017fully}, which in particular gives us a $2+\eps \leq 3$-approximation of $\mu(G_i)$ in worst-case update time $\tilde{O}_{\eps}(1)$, again fitting within our time budgets.
	    Now, during phase $i$, we spend the time $t_q$ for the $\eps$-AMM query subroutine of $\calA$, so as to finish computing $M_i$.
	    The amount of time spent per update to achieve this goal is 
	    at most $\frac{\tilde{O}_{\eps}(\mu(G_i))}{\lfloor\eps\cdot \mu(G_i)/10\rfloor} = \tilde{O}_{\eps}(1),$
	    again fitting within our time updates.
	    We now describe and analyze the matchings maintained by this algorithm (these are not always $M_i$).
	    
	    By \Cref{AMM-stable}, we need to provide a matching $M'_{i+1}$ at the start of each phase $i+1$ which is an $O(\eps)$-AMM in $G_{i+1}$, thus guaranteeing that the non-deleted edges of $M'_{i+1}$ remain an $O(\eps)$-AMM. Reparameterizing appropriately will then yield the desired result. It remains to define our matchings $M'_{i}$.
	    Using our estimate of $\mu(G)$ obtained by the dynamic fractional matching, we test whether $\mu(G_i) \in [1/\eps, 10/\eps]$. If this is the case, then $M'_{i+1}$ is obtained by querying the AMM algorithm $\calA$ at the beginning of phase $i+1$, in time $O_{\eps}(1)$. (This relied on $\mu(G_{i+1})\leq \mu(G_i) + \ell_i \leq \mu(G_i)\cdot (1+\eps) = \tilde{O}_{\eps}(1)$.) By the properties of $\calA$, the matching $M'_{i+1}$ is an $\eps$-AMM in $G_{i+1}$ w.h.p.
	    Now, if conversely $\mu(G_i) \geq 10/\eps$, then we have that 
	    $$\frac{1}{2}\cdot \mu(G_i) \leq (1-\eps)\cdot \mu(G_i) \leq \mu(G_i)-\ell_i \leq \mu(G_{i+1}).$$
	    Now, $M_i$, is an $\eps$-AMM in $G_i$, which is obtained from $G_{i+1}$ by at most $\eps\cdot \mu(G_i)\leq 2\eps\cdot \mu(G_{i+1})$ updates. Therefore, by \Cref{AMM-stable} $M_i$ is a $12\eps$-AMM in $G_{i+1}$. We therefore take $M'_{i+1}$ to be $M_i$. Reparameterizing $\eps$ appropriately, the lemma follows.
	\end{proof}

	\section{Conclusion and Future Directions}\label{sec:conclusion}
	
	We presented the first dynamic matching (size estimation) algorithm breaking the approximation barrier of $2$ in polylogarithmic update time.
	While this presents a major advance in our understanding of the dynamic matching problem, many questions remain. We mention a few such questions which we find particularly intriguing.
	
	\paragraph{Explicit Fast Matching.} In our work we show how to maintain a better-than-two approximate estimate of the maximum matching size. Can one also maintain an explicit matching of similar approximation ratio within the same time bounds?
	
	\paragraph{Better approximation in $o(n)$ update time?} Known conditional impossibility results rule out an exact algorithm with $n^{1-\Omega(1)}$ update time \cite{abboud2014popular,henzinger2015unifying,dahlgaard2016hardness}, but the best approximation ratios currently known are $\frac{3}{2}+\eps$ \cite{bernstein2015fully,bernstein2016faster,kiss2022improving,grandoni2022maintaining}. Can one do better in $o(n)$ time?
	On the flip side, can we show any (conditional) hardness of \emph{approximate} dynamic matching, for any approximation ratio?
	
	\paragraph{Unconditional impossibility results.} 
	With this work we bring dynamic matching with better-than-two approximation into the polylogarithmic update time regime---the range where \emph{unconditional} impossibility are known for numerous data structures and dynamic algorithms. 
	Can such unconditional impossibility results be proven for (approximate) dynamic matching?

 \paragraph{Acknowledgements.} We thank the anonymous reviewers for helpful comments.

\appendix
\section*{APPENDIX}
	
	\section{Proofs of basic building blocks}\label{appendix:previous-blocks}
	
	Here we substantiate some key propositions implied by prior work.
	We stress that we provide proofs mostly for completeness, due to our propositions being slight variants or being differently organized than their previous counterparts. That is, we do not claim novelty of the underlying ideas of this section.

	\subsection{Proof of Proposition~\ref{reduction}}

	A key component of \Cref{reduction} is the following vertex sparsification technique for dynamic settings by Kiss \cite{kiss2022improving}, adapted from such a vertex sparsification of Assadi et al.~\cite{assadi2016maximum} in the context of stochastic matching.
	
	\begin{restatable}{prop}{vertexsparsification}\label{vertex-sparsification}
		There exists a randomized algorithm which for each update to $G$ makes an update to $O\left(\frac{\log^2n}{\eps^3}\right)$ contracted subgraphs, such that w.h.p.~throughout any (possibly adaptively generated) update sequence, one subgraph $G'$ has a matching of cardinality $\mu(G)\cdot (1-O(\eps))$ and nodeset of size $n'\leq \mu(G)/\eps$.
		Moreover, any matching $M'$ in $G'$ can be transformed into a matching in $G$ of cardinality $|M'|$ in time $O(|M'|)$. For any matching $M'$ in any $G'$ undergoing edge updates we can maintain a matching of cardinality $|M'|$ in $G$ with $O(1)$ worst-case update time.
	\end{restatable}
	\begin{proof}
	Consider a random graph $G'$ obtained by hashing each node into one of $k/\eps$ buckets, for some integer $k$, and contracting all nodes that are hashed into the same bin.
	That is, two contracted nodes neighbor in $G'$ if their corresponding bins contain neighboring nodes in $G$.
	By storing for each edge $e$ in $G'$ a list of edges inducing $e$, we can easily transform a matching $M'$ in $G'$ to a matching in $G$ of the same cardinality in time $O(|M'|)$. The majority of this proof is thus dedicated to showing that $O\left(\frac{\log n}{\eps^2}\right)$ such contractions for each value $k=\lceil (1+\eps)^i\rceil$ with $i\in [\log_{1+\eps}(n)]\subseteq \left[O\left(\frac{\log n}{\eps}\right)\right]$ suffice to guarantee that one of these $G'$ contains a matching of cardinality at least $\mu(G)\cdot (1-3\eps)$.
		
	Fix an integer $i$ and $k=\lceil (1+\eps)^i\rceil \leq n$. 
	Fix a matching $M$ in $G$ of cardinality $|M|\leq k$. The probability that a vertex $v$ incident on some edge of $M$ is contracted into a separate bin than the other $2|M| - 1$ endpoints can be expressed as follows:
	\begin{equation}
	(1 - 1/(k/\epsilon-1))^{2|M|-1} \geq \left(1 - \frac{2 \cdot \epsilon}{k} \right)^{2 \cdot k} \geq \nonumber (1 - 5 \cdot \epsilon).
	\end{equation}
	Thus, by linearity, the number $X$ of such endpoints of edges of $M$ satisfy that $\E[X] \geq (1-5\epsilon)\cdot2|M|$. Observe that $X$ is the sum of negatively associated random variables, by \cite{DubhashiR98}, since the hashing of vertices is equivalent to the folklore balls and bins experiment, so by standard Chernoff Bounds, 
	$$\Pr[X \leq 2 \cdot |M|\cdot (1-6\eps)]\leq \exp\left(-\Theta(\eps^2 |M|)\right).$$
	
    If at least $2|M|\cdot(1 - 6 \cdot \epsilon)$ endpoints of $M$ are hashed to unique vertices then at least $|M| -  2|M| \cdot 6 \epsilon \geq |M| \cdot (1- 12\eps)$ edges of $M$ had both of their endpoints assigned to unique vertices in $G'$ hence are present in $G'$.
	
	We say that the contraction is \emph{bad} if for \emph{some} matching $M$ of cardinality in the range $[k,k(1+\eps)+1]$ if the number of edges of $M$ that are not present in $G'$ is lesser than $|M|\cdot (1-12\eps)$. Otherwise, it is \emph{good}.
	Now, there are $\sum_{i=k}^{k(1+\eps)+1}{n \choose i}\leq k\eps\cdot n^{k(1+\eps)+1}\leq n^{k(1+\eps)+2}$ possible matchings of size $|M|\in [k,k(1+\eps)]$. 
	Therefore, by randomly contracting the graph for range $[k,k(1+\eps)+1]$ some $\frac{C\log n}{\eps^2}$ many times, for a sufficiently large $C$, we have that the probability that all contractions for range $[k,k(1+\eps)+1]$ are bad is
	\begin{align*}
	\Pr[\textrm{all contractions are bad}] & \leq n^{k(1+\eps)+2} \cdot \exp\left(-\Theta(\eps^2 k) \cdot \frac{C \log n}{\eps^2}\right) \leq n^{-3}.
	\end{align*}
	Therefore, taking union bound over the $\log_{1+\eps}(n)$ possible value of $k$, we find that with high probability, each range $[k,k(1+\eps)+1]$ has some good contraction.
	
	We conclude that, w.h.p., among the $O\left(\frac{\log^2n}{\eps^3}\right)$ contracted graphs, there exists a good contraction for every $k=\lceil (1+\eps)^i\rceil$, and in particular for $k\leq \mu(G) \leq k(1+\eps)+1$. That is, one of the contracted graphs contains a large matching, $\mu(G')\geq \mu(G)\cdot (1-12\eps)$, and has few nodes, $n'\leq k/\eps \leq \mu(G)/\eps$, as desired.
	\end{proof}

	We now proceed towards proving \Cref{reduction}, restated below for ease of reference. 	
	\reduction*
	
	\begin{proof}
	Let $\epsilon' = \alpha' \cdot \epsilon \cdot 2$ (here $\alpha'$ is some $O(1)$ factor). Using the algorithm described by \Cref{vertex-sparsification} we can generate $T = \tilde{O}_\epsilon(1)$ graphs $G_i : i \in [T]$ with the following properties: A) $\mu(G_i) \leq \mu(G)$ for all $i \in [T]$, B) There is an $i \in [T]$ satisfying that $\mu(G_i) \geq (1 - \epsilon' ) \cdot \mu(G)$ and $\mu(G_i) \geq n \cdot \epsilon'$, C) All sub-graphs $G_i$ undergo a single update when $G$ undergoes an update. 
	
	Our algorithm proceeds as follows: on all $T$ generated sub-graphs we run algorithm $\calA$ at all times. Furthermore, on each sub-graph we maintain an $O(1) = \alpha'$-approximate estimate on the maximum matching size $\tilde{\mu_i}$ using algorithms from literature (randomized against an adaptive adversary) in $\tilde{O}_\epsilon(1)$ worst-case time. For all sub-graphs we monitor the relationship of $\tilde{\mu_i}$ and $|V_i|$. If $\tilde{\mu_i}$ increases above the threshold of $|V_i| \cdot \epsilon$ we start a run of the query algorithm on $G_i$ returning us an $\alpha$-approximate estimate of $\mu(G_i)$ which will define $\nu_i'$. We distribute the work of this query over $|V_i| \cdot (\epsilon)^2$ updates and re-initiate the query every $|V_i| \cdot (\epsilon)^2$ updates. The matching size queries of $G_i$ always run on the state of $G_i$ at the start of the query (even though $G_i$ undergoes updates during it's run). If $\tilde{\mu}_i$ decreases bellow the threshold of $|V_i| \cdot \epsilon$ we stop the querying process and set $\nu_i' = 0$. Note that at initialization we just set $\nu_i' = \mu(G_i)$ for all $i \in [T]$ statically.
	
	At all times we maintain the output $\max_{i \in R} \nu_i'$, the maximum of our matching size estimates.
	    
	  \begin{algorithm}[H]
		\caption{Vertex set sparsification}
		\label{alg:reduction}
		\begin{algorithmic}[1]
		    \State Initialize $\nu_i' = \mu(G_i)$
            \State Maintain contracted sub-graphs $G_i$ and $\alpha'$-approximate matching size estimates $\tilde{\mu}_i$
            \State Run algorithm $\calA$ on every $G_i$
            \For{$i \in [T]$}
                \If {$\tilde{\mu}_i$ becomes at least $|V_i|\cdot \epsilon$}
                    \State Initiate a matching size query of $G_i$ in $O(t_q)$ time on the current state of $G_i$
                    \State Distribute the work over the next $|V_i| \cdot \epsilon^2$ updates 
                    \State Repeatedly recompute distributed over every $|V_i| \cdot \epsilon^2$ updates
                    \State Let $\nu_i'$ be the latest finished estimate
                \EndIf
                \If {$\tilde{\mu}_i$ reduces below $|V_i|\cdot\epsilon$}
                    \State Terminate the querying process of $\mu(G_i)$
                    \State Set $\nu_i' \leftarrow 0$
                \EndIf
            \EndFor
            \State At all times return $\max_{i \in [T]}  \nu_i'$
		\end{algorithmic}
	\end{algorithm}	
	
	We first discuss the update time of Algorithm~\ref{alg:reduction}. The maintenance of the $T$ contracted sub-graphs and matching size estimates $\tilde{\mu}_i$ takes $\tilde{O}_\epsilon(1)$ w.c.~time. Running algorithm $\calA$ on each of the contracted sub-graphs takes update time $\tilde{O}_\epsilon(t_u)$ (and is worst case if $\calA$ has worst-case update time). A matching size query will only be initiated and run on contracted sub-graph $G_i$ if $\tilde{\mu}_i\geq \mu(G_i)/\alpha' \geq |V_i| \cdot \epsilon$, that is if $\mu(G_i) \geq |V_i| \cdot \epsilon/2$. Each re-computation of the estimate $\nu_i'$ will be distributed over some $|V_i| \cdot \epsilon^2$ updates, that is, it will take $O(t_q/(|V_i|\cdot \epsilon^2) = O_\epsilon(t_q / |V_i|)$ worst-case time. Finding and returning the maximum matching size estimate $\nu_i'$ takes $\tilde{O}_\epsilon(1) = O(T)$ time. Therefore, the total update time of the algorithm is $\tilde{O}_\epsilon(t_u + t_q \cdot \beta/n)$ and is worst-case if $\calA$ has worst-case update time. Furthermore, all components of the algorithm but $\calA$ are randomized against an adaptive adversary..
	
	It remains to argue that the algorithm maintains $\nu'$ such that $\nu' \leq \mu(G) \leq \nu' \cdot (\alpha + O(\epsilon))$ at all times. 
	Say that $G_i$ is a 'successful' contraction if $G_i$ satisfies property B). By \Cref{vertex-sparsification}, w.h.p., there is a successful contraction at all times, at time point $\tau_1$ let that contraction be $G_i$. We will separate two instances: 
	
	i) Throughout the run of the algorithm at all times it held that $\tilde{\mu}_i \geq |V_i| \cdot \epsilon$: The algorithm has ran the matching size query sub-routine on $G_i$ after every $|V_i| \cdot \epsilon^2$ edge updates. Let $G^{\tau_0}_i$ be the past state of the graph $G_i$ when the algorithm started calculating the current estimate $(\nu^{\tau_1}_i)'$. By the scheduling of this calculation we know that $\tau_0 \geq \tau_1 - \epsilon^2 \cdot |V_i|$. Hence, $\mu(G^{\tau_0}_i) \geq \mu(G^{\tau_1}_i) - \epsilon^2 \cdot |V_i|$, where $\mu(G^{\tau_1}_i) \geq \mu(G) \cdot (1 - \epsilon')$ and $\mu(G^{\tau_1}_i) \geq |V|_i \cdot \epsilon'$. Hence, $(\nu^{\tau_1}_i)' \cdot \alpha \cdot (1 + O(\epsilon)) \geq \mu(G)$.
	
	ii) At time $\tau_1$ $G_i$ is a successful contraction but at some prior point during the run of the algorithm $\tilde{\mu}_i$ became less than $|V_i| \cdot \epsilon$: we know that at some point $\tau_0$ prior to $\tau_1$ $\tilde{\mu}_i$ must have increased above $|V_i| \cdot \epsilon$. Define the state of $G_i$ at the two time points as $G^{\tau_0}_i$ and $G^{\tau_1}_i$ respectively. As at $\tau_1$ $G^{\tau_1}_i$ is a successful contraction we know that $\mu(G^{\tau_1}_i) \geq |V_i| \cdot \epsilon'$. When $\tilde{\mu}_i$ crossed the threshold at $\tau_0$ it held that $\tilde{\mu}_i = |V_i| \cdot \epsilon$ that is $\mu(G_i^{\tau_0})_i \leq |V_i| \cdot \epsilon \cdot \alpha$. As per each update the maximum matching size may only change by $1$ we have that $\tau_1 - \tau_0 \geq |V_i| \cdot \epsilon \cdot \alpha$. Hence, by time $\tau_1$ the algorithm already had an updated estimate of $\nu_i'$ (that is one calculated in the previous $\epsilon^2 \cdot |V_i|$ updates such that $\mu(G_i) \geq |V_i| \cdot \epsilon$ during these updates). Here we can refer back to the previous case (pretending the algorithm initialized at $\tau_0$).
\end{proof}

\subsection{Proof of Proposition~\ref{poor-approx-maximal-matching=many-aug-paths}}
	
	We now give a proof extending standard arguments that small maximal matchings contain many length-three augmenting paths to showing that small $\eps$-AMM likewise contain many such paths.
	
	\shortaug*
	\begin{proof}
		The above bound for $\epsilon=0$ is well-known (see, e.g., \cite{konrad2012maximum}). We reduce to this case by removing the at most $\eps \cdot \mu(G)$ nodes in $V\setminus V(M)$ needed to make $M$ maximal.
		This yields a graph $G'$ with $\mu(G')\geq \mu(G)\cdot (1-\eps)$, and therefore
		$$|M|\leq \left(\frac{1}{2}+c\right)\cdot \mu(G)\leq \frac{\frac{1}{2}+c}{1-\eps}\cdot \mu(G') \leq \left(\frac{1}{2}+c+\eps\right)\cdot \mu(G'),$$ 
		Consequently, by the special case of this proposition with $\eps=0$, we have that the maximum number of disjoint 3-augmenting paths that $M$ admits in $G'$ (and hence also in $G$) is at least 
		\begin{align*}
		\left(\frac{1}{2}-3(c+\eps)\right)\cdot \mu(G') & \geq \left(\frac{1}{2}-3(c+\eps)\right)(1-\eps)\cdot \mu(G) \geq \left(\frac{1}{2}-3(c+\eps)-\frac{\eps}{2}\right) \cdot \mu(G). \qedhere
		\end{align*}
	\end{proof}

	\section{Proof of Lemma~\ref{lm:sublinear:extension}}
	\label{appendix:sublinear:extension}

	Our proof of Lemma~\ref{lm:sublinear:extension} is a minor modification of the argument from Section 5 of \cite{behnezhad2022time}. We claim no novelty for this proof. To make our notations consistent with the ones used by  \cite{behnezhad2022time}, we will focus on an $n$-node graph $G=(V,E)$ (different from our dynamic input graph).
	Let $\pi$ be a permutation of the edges of graph $G = (V,E)$. Let $GMM(G,\pi)$ stand for the output of the greedy maximum matching algorithm when run on graph $G$ with edge ordering $\pi$.

	\subsection{Building blocks}
	
 \Cref{lm:sublinear:extension:support:1} is explicitly concluded by \cite{behnezhad2022time}, whereas \Cref{lm:sublinear:extension:support:2} is a slight modification of a construction appearing in Section 5 of \cite{behnezhad2022time} we need to fit our arguments.

	\begin{lem}
	\label{lm:sublinear:extension:support:1}
	There is a randomized algorithm that in $\tilde{O}(|E|/|V|)$ expected time returns the matched status of a random $v$ under $GMM(G, \pi)$, for random $\pi$. This algorithm relies on list access to the edges of $G$.
	\end{lem}
	
	In order to prove \Cref{lm:sublinear:extension} we have to work with adjacency matrix queries. Based on a slight modification of Section 5 of \cite{behnezhad2022time} we can derive the following tool for this purpose. 
	
	\begin{lem}
	
	\label{lm:sublinear:extension:support:2}
	Let $\delta\in (0,1/2)$. For a given $n$-node graph $G = (V,E)$ there exists a supergraph $H = (V_H, E_H)$ of $G$ (i.e., $V_H\supseteq V$ and $E_H\supseteq E$)  satisfying the following:
	\begin{itemize}
	    \item $|E_H| = \Theta_\delta(n^2)$.
	    \item $|V_H| = \Theta_\delta(n^2)$.
	    \item At most $\delta \cdot n$ nodes of $V$ are matched to nodes in $V_H \setminus V$ by $GMM(H,\pi)$, w.h.p.~over $\pi$.
	    \item $GMM(H,\pi)\cap E$ is a maximal matching in $G[V\setminus V_{slack}]$, where $V_{slack}\subseteq V$ are nodes in $V$ that are matched to nodes in $V\setminus V_H$.
	    \item Any adjacency list query to $E_H$ (querying the $i$-th neighbour of a vertex according to some ordering of neighbours) can be implemented using one adjacency matrix query to $E$ (querying the existence of any edge $(u,v)$).
	\end{itemize}
	
	\end{lem}
	
	Informally, the main change in our construction compared to that of \cite{behnezhad2022time} is that our construction will allow us to argue that the random matching in the constructed graph $H$ is, w.h.p., a maximal matching after ignoring a small set of nodes. In contrast, the construction in \cite{behnezhad2022time} resulted in an ``expected'' version of this guarantee. As the high-probability bounds will simplify our discussion later, we modify this construction below. 
	The second change we make is in externalizing the fact that the matching computed this way is maximal, rather than 2-approximate, as stated in \cite{behnezhad2022time}. We now turn to proving the above lemma.
	
	\begin{proof}[Proof of \Cref{lm:sublinear:extension:support:2}]
	    
	    The node-set of $H$ is $V_H := V \cup V^* \cup (W_1 \ldots, W_n) \cup (U_1, \ldots, U_n)$, where $V = \{v_1, v_2, \ldots, v_n\}$ (note that $G = (V, E_H[V])$), $V^* = \{v^*_2, v^*_2, \ldots, v^*_2\}$, $W_i = \{w^1_i, w^2_i, \ldots w^n_i \}$, and the set $U_i = \{u^1_i, u^2_i, \ldots, u^s_i\}$ is of size $s := 10 n /\delta$ for all $i \in [n]$. To specify the edge-set $E_H$, we now define the {\em ordered} adjacency list for every node $v \in V_H$.
	    \begin{itemize}
	        \item Every node $v_i \in V$ has degree exactly $n$: For any $j \in [n]$, if $(v_i, v_j) \in E$ then the $j^{th}$ neighbor of $v_i$ is the node $v_j \in V$, otherwise it is the node $v^*_j \in V^*$. 
	        \item Every node $v^*_i \in V^*$ has degree exactly $n+s$: For any $j \in [n]$, if $(v_i, v_j) \in E$ then the $j^{th}$ neighbor of $v^*_i$ is the node $w^j_i \in W_i$, otherwise it is the node $v_j \in V$. Furthermore, for all $j \in [s]$, the $(n+j)^{th}$ neighbor of $v^*_i$ is the node $u^j_i \in U_i$.
	        \item Each node in $U_j$, for any $j \in [n]$, has only one neighbor (which is $v^*_j$).
	        \item Node $w^j_i$ may have degree at most one: if $(v_i,v_j) \in E$ then $w^j_i$ is a neighbour of $v^*_i \in V^*$, otherwise it is an isolated vertex of $H$.
	    \end{itemize}
    
    \noindent Note that $|V_H| = 2n + n^2 + n s = \Theta(n^2/\delta)$ and that similarly $|E_H|=n^2 + |E| + ns = \Theta(n^2/\delta)$. Furthermore, from the above discussion it is immediate that an adjacency list query to $E_H$ (i.e., querying for the $j$-th neighbor of a vertex) can be implemented using at most one adjacency matrix query to $E$. It remains to prove the remaining two properties of $H$.
    
    To this end, recall that $V_{slack}$ denotes the set of vertices in $V$ matched to $V$* nodes. Then, by maximality of $GMM(H, \pi)$, we have that $GMM(H,\pi)\cap E$ is indeed maximal matching of $G[V \setminus V_{slack}]$. 
    We now turn to bound $|V_{slack}|$.
    To this end, we say a node $v^*\in V^*$ is \textit{occupied} if its earliest edge in $\pi$ has its other endpoint in $W_i$ or $U_i$. Trivially, such an occupied vertex $v^* \in V^*$ is matched to a vertex of $U_{v^*} \cup W_{v^*}$ under $GMM(H, \pi)$. 
    The following simple claim, which follows by a Chernoff bound together with the simple observation that it is unlikely for a node in $V*$ to be matched in $V$ (and thus contribute to $|V_{slack}|$). 
    
    \begin{claim}
    
    Let $\pi$ be a uniformly random permutation of $E_H$. Let $X_{v^*} : v^* \in V^*$ represent the indicator variable of $v^*$ being occupied and $X_O = \sum X_{v^*}$. Then $X \geq n \cdot (1 - \delta)$ w.h.p.
    \end{claim}
    
    \begin{proof}
    Note that each $v^*\in V^*$ has at most $n$ edges with vertices of $V'$ and has at least $10n/\delta$ edges with vertices in $U_{v^*}$ and $W_{v^*}$. Therefore, 
    $$\E[X_{v^*}]=\Pr(X_{v^*} = 1) \geq \frac{n \cdot 10/\delta}{n \cdot (10/\delta + 1)} = 1 - \frac{1}{10/\delta + 1} \geq  1 - \delta/10.$$
    On the other hand, the variables $\{X_{v^*} \mid v\in V\}$ are independent binary variables. Therefore, by Chernoff's bound, we have that
    \begin{eqnarray}
    \Pr(X_O \leq n \cdot (1 - \delta)) & \leq & \Pr\left(X_O \leq n \cdot (1 - \delta/10) - \frac{n \cdot \delta}{2}\right)\nonumber \\
    & \leq & \Pr\left(X_O \leq \E[X_O] - \E[X_O]  \cdot \frac{\delta}{2}\right) \label{eq:lm:sublinear:extension:support:2:1}  \\
    & \leq & 2 \cdot \exp \left(- \frac{(\delta/2)^2 \cdot \E[X_O]}{3}\right) \label{eq:lm:sublinear:extension:support:2:2} \\
    & \leq & n^{- \Theta(1)} \label{eq:lm:sublinear:extension:support:2:3}
    \end{eqnarray}

    Inequality~\ref{eq:lm:sublinear:extension:support:2:1} follows from the fact that $n \geq E[X_O] \geq n \cdot (1 - \delta)$. Inequality~\ref{eq:lm:sublinear:extension:support:2:2} is an application of Chernoff's bound. Inequality~\ref{eq:lm:sublinear:extension:support:2:3} follows as long as $n \cdot \delta^2 \in \Omega(\log(n))$.
	\end{proof}
	The above claim completes the proof of the last requirement of Lemma~\ref{lm:sublinear:extension:support:2}.
	\end{proof}

	\subsection{The algorithm}
	We now introduce the algorithm that will build on the previous two lemmas and inform the proof of \Cref{lm:sublinear:extension}, given in \Cref{alg:lm:sublinear:extension}.
	Recall that we wish to estimate the number of edges in some input matching, which here, to avoid confusion, we denote by $M^*$, that are both matched in some maximal matching in $G$.
	
	 Let $H = (V_H, E_H)$ be the supergraph of $G$ defined by \Cref{lm:sublinear:extension:support:2} of $G$ with $\delta = \epsilon^2/8$. For a permutation $\pi$ of $E_H$, define $M'(\pi) := \text{GMM}(H, \pi) \cap (V \times V)$ to be the set of edges in $\text{GMM}(H, \pi)$ both of whose endpoints are in $V$. Let $M(\pi)$ be a maximal matching in $G$ that is obtained by augmenting $M'(\pi)$, i.e., we start with $M = M'(\pi)$, visit the edges $e \in E$ in an arbitrarily fixed order, and obtain the matching $M(\pi)$ by greedily  adding as many edges to $M$ as possible. Note that $M'(\pi) \subseteq M(\pi) \subseteq E'$. Note also that $M(\pi)$ will be the maximal matching that \Cref{lm:sublinear:extension} refers to as $M'$. We now slightly overload our notations and let $k_{M'(\pi)}$ denote the number of edges in $M^*$ both of whose endpoints are matched in $M'(\pi)$.\footnote{Recall that in the statement of~\Cref{lm:sublinear:extension} we defined the notation $k_{M'}$ only if $M'$ is a matching in $G'$, which is not the case with $M'(\pi)$. Nevertheless, for ease of exposition, we use the notation $k_{M'(\pi)}$.} 
	
	\begin{algorithm}[H]
		\caption{Extended Sub-Linear Algorithm}
		\label{alg:lm:sublinear:extension}
		\begin{algorithmic}[1]
		    \If {$|M^*| \leq \epsilon^2 \cdot n$}
		    \State \textbf{Return} $\kappa = 0$
		    \EndIf
		    \State (Implicitly) construct $H=(V_H,E_H)$ as in \Cref{lm:sublinear:extension:support:2} with $\delta = \epsilon^2/8$
		    \State Sample a permutation $\pi$ of $E_H$ uniformly at random
			\State $L \leftarrow \frac{10^5 \cdot \log(n)}{\epsilon^5} $
			\State Sample $L$ edges $e_1,\dots, e_L\in M^*$ uniformly at random with replacement 
			\State Let $X_i$ be one if both endpoints of edge $e_i$ are matched by $GMM(H,\pi)$ and $X = \sum_i X_i$
			\State \textbf{Return} $\kappa = \frac{X  \cdot |M|}{L} - \frac{n \cdot \epsilon^2}{2}$
		\end{algorithmic}
	\end{algorithm}

	\begin{claim}\label{sublinear-expectation}
	
	Algorithm~\ref{alg:lm:sublinear:extension} can be implemented in time $\tilde{O}_\epsilon(n)$ in expectation.
	
	\end{claim}
	
	\begin{proof}
	The construction of $H$ is implicit, and as such takes no time.
	Let $T_H(v,\pi)$ stand for the time it takes to calculate the matched status of vertex $v \in V_H$ in $GMM(H,\pi)$ using the algorithm of \cite{behnezhad2022new}. By Lemma~\ref{lm:sublinear:extension:support:1} we have that $\underset{v \sim V_H}{\E}[T_H(v,\pi)] = \tilde{O}_\epsilon(|E_H|/|V_H|) = \tilde{O}_\epsilon(1)$. Therefore, since the endpoints of the sampled edges $S=\bigcup_{i=1}^L {e_i} \subseteq V_H$ are a subset of of vertices of cardinality $|S| \geq \epsilon^2 \cdot n$, and since $|V_H|=\Theta_{\eps}(n^2)$ we have the expected time to calculate their matched status (using adjacency matrix queries, using the construction of $H$) is
	\begin{align*}
	\underset{v \sim S}{\E_H}[T_H(v, \pi)] & \leq  \underset{v \sim V_H}{\E}[T_H(v, \pi)] \cdot \frac{|V_H|}{|S|} \leq \tilde{O}_\epsilon(n)\qedhere
	\end{align*}
	\end{proof}
	
	We now argue that \Cref{alg:lm:sublinear:extension} provides a good approximation of the number of nodes in $M^*$ both of whose endpoints are matched by $GMM(H,\pi)$. But first, we recall the basic Chernoff bounds that we will rely on here.
	
	\begin{lem}
	\label{lm:chernoff}
	Chernoff bound: Let $X$ be the sum of independently distributed (or negatively associated) random variables $X_1, \dots , X_m$ with $X_i \in [0,1]$ for each $i \in [m]$. Then for all $\delta \in (0,1)$:
	
	\begin{equation}
	    \Pr(|X - E[X]| \geq \delta \cdot E[X]) \leq 2 \cdot \exp\left(-\frac{\delta^2 \cdot E[X]}{3}\right). \nonumber
	\end{equation}
	
	\end{lem}
	
	\begin{lem}\label{correct-whp}
		W.h.p., The output $\kappa$ of \Cref{alg:lm:sublinear:extension} satisfies $$\kappa_{M(\pi)} \geq \kappa \geq \kappa_{M(\pi)} - n \cdot \epsilon^2.$$
	\end{lem}
	
	\begin{proof}
	First, by \Cref{lm:sublinear:extension:support:2}, we have that w.h.p., the set of nodes $V_{slack} \subseteq V$ that are matched to nodes in $V_H \setminus V$ have cardinality at most $|V_{slack}|\leq n\epsilon^2/8$.
	Moreover, $GMM(H, \pi) \cap E$ is a maximal matching in $G[V \setminus V_{slack}]$.
	
    Observe that whenever $\kappa = 0$ is returned by the algorithm due to $|M^*|$ being small, the algorithm returns a trivially correct solution. As $M'(\pi)$ is an $\epsilon^2/8$-AMM, we  conclude that:
    
    \begin{equation}
    \label{eq:1:lm:sublinear:extension}
    |M'(\pi)| \leq |M(\pi)| \leq |M'(\pi)| + \frac{n \cdot \epsilon^2}{8}.
    \end{equation}
    
    Define $M^*_H$ to be the set of edges of $M^*$ such that both of their endpoints are matched by $GMM(H, \pi)$. By the guarantees of the construction of $H$ we know that there can be at most $n \cdot \epsilon^2/8$ vertices of $V$ matched by an edge not in $M'(\pi)$. Therefore, 
    
    \begin{equation}
    \label{eq:2:lm:sublinear:extension}
    |M^*_H| \geq \kappa_{M'(\pi)} \geq |M^*_H| - \frac{n \cdot \epsilon^2}{8}.
    \end{equation}
    
    Note that using the Algorithm~\ref{alg:lm:sublinear:extension} is sampling from edges of $M^*$ and determining if they are in $M^*_H$ (hence approximating $\kappa_{M^*_H}$). Specifically, by inequalities (\ref{eq:1:lm:sublinear:extension}) and (\ref{eq:2:lm:sublinear:extension}), we get the following.
    \begin{eqnarray}
    \kappa_{M^*_H} & \in \left[\kappa_{M'(\pi)} \pm  \frac{n \cdot \epsilon^2}{8} \right] \subseteq  \left[\kappa_{M(\pi)} \pm  \frac{n \cdot \epsilon^2}{8} \pm |M(\pi)| - |M'(\pi)| \right] \subseteq  \left[\kappa_{M(\pi)} \pm \frac{n \cdot \epsilon^2}{4} \right]. \nonumber
    \end{eqnarray}

    \medskip
    We will argue that with high probability $\frac{X \cdot |M|}{L} \in \left[\kappa_{M^*_H} \pm \frac{n \cdot \epsilon^2}{8} \right]$, dependent on the randomization of $M^*_L$. Observe that $X_i$ are independently distributed random variables taking values in $[0,1]$ and $X$ is a binomial variable with parameters $k, |\kappa_{M^*_H}|/|M|$. We will consider two cases:
    
    \medskip
    \noindent 
    {\bf Case (A):} $\kappa_{M^*_H} \leq \frac{n \cdot \epsilon^3}{8}$. In this case, we derive that
    \begin{eqnarray}
    \Pr\left(\frac{X \cdot |M|}{L} \notin  \left[\kappa_{M^*_H} \pm \frac{n \cdot \epsilon^2}{8} \right]\right) & = & \Pr\left(\frac{X \cdot |M|}{L} \geq \kappa_{M^*_H} + \frac{n \cdot \epsilon^2}{8}\right) \nonumber \\
    & \leq & \Pr\left(\frac{X \cdot |M|}{L} \geq \frac{n \cdot \epsilon^2}{8}\right) \nonumber \\
    & = & \Pr\left(B(L,\kappa_{M^*_H}/|M|)\geq \frac{n \cdot \epsilon^2}{8}\right) \label{eq:3:lm:sublinear:extension}\\
    & \leq & \Pr\left(B(L,\epsilon) \geq \frac{n \cdot \epsilon^2}{8}\right)  \label{eq:4:lm:sublinear:extension}\\
    & \leq & \Pr\left(B(L,\epsilon) \geq 2 \cdot \E[B(L,\epsilon)]\right) \nonumber  \\
    & \leq & 2 \cdot \exp \left(-\frac{\E[B(L,\epsilon)]}{3}\right) \label{eq:5:lm:sublinear:extension} \\
    & \leq & n^{-\Theta(1)}. \nonumber
    \end{eqnarray}
    
    In the above derivation,~(\ref{eq:3:lm:sublinear:extension}) holds as $\kappa_{M^*_H}/|M| \leq \epsilon$ (as otherwise $\kappa = 0$ would have been returned by the algorithm), and~(\ref{eq:4:lm:sublinear:extension}) is true assuming $n \cdot \epsilon^2/8 \geq 2 \cdot L \cdot \epsilon = \tilde{O}(1)$. Finally,~(\ref{eq:5:lm:sublinear:extension}) follows from Chernoff bound (Lemma~\ref{lm:chernoff}) on a binomial random variable. 
    
    \medskip 
    \noindent
    {\bf Case (B):} $\kappa_{M^*_H} \geq n \cdot \frac{\epsilon^3}{8}$. In this case, we derive that
	
	\begin{eqnarray}
	\Pr\left(\frac{X \cdot |M|}{L} \notin  \left[\kappa_{M^*_H} \pm \frac{n \cdot \epsilon^2}{8} \right]\right) & = & \Pr\left(|X - E[X]|\geq \frac{n \cdot \epsilon^2}{8} \cdot \frac{L}{|M|}\right) \nonumber \\
	& = & \Pr\left(|X -  \E[X]| \geq \E[X] \cdot \frac{n \cdot \epsilon^2}{8 \cdot \kappa_{M^*_H}}\right) \nonumber \\
	& \leq & \Pr\left(|X -  \E[X]| \geq \E[X] \cdot \frac{n \cdot \epsilon^2}{8 \cdot n}\right) \nonumber \\
    & \leq & 2 \cdot \exp\left(-\frac{\left( \epsilon^2/8\right)^2 \cdot \E[X]}{3}\right)	\label{eq:6:lm:sublinear:extension} \\
    & = & \exp \left(- \frac{L \cdot \kappa_{M^*_H} \cdot \epsilon^4}{|M| \cdot 194}\right) \label{eq:7:lm:sublinear:extension} \\
    & \leq & n^{-\Theta(1)}. \nonumber
	\end{eqnarray}
	In the above derivation,~(\ref{eq:6:lm:sublinear:extension}) follows from Chernoff bound (Lemma~\ref{lm:chernoff}), and~(\ref{eq:7:lm:sublinear:extension}) holds due to on our assumptions on $|M|$ and $\kappa_{M^*_H}$. Therefore, with high probability $X \cdot |M|/k \in [\kappa_{M^*_H} \pm \epsilon^2 \cdot n/8] \in [\kappa_{M(\pi)} \pm \epsilon^2 \cdot n \cdot (1/8 + 1/4)]$ (recall that $\kappa_{M(\pi)} = \kappa_{M}$). This implies that $\kappa \leq \kappa_{M} + \epsilon^2 \cdot n \cdot (1/8 + 1/4 - 1/2) \leq  \kappa_{M} $ and $ \kappa  \geq  \kappa_{M}  - \epsilon^2 \cdot n \cdot (1/8 + 1/4 - 1/2) \geq  \kappa_{M}  - \epsilon^2 \cdot n$. \end{proof}

	Having concluded that \Cref{alg:lm:sublinear:extension} can be implement in low expected time, and is correct w.h.p., we are now ready to prove \Cref{lm:sublinear:extension}, restate below for ease of reference. (Note that here $G=(V,E)$ are renamed $G'=(V',E')$, and $n'$ and $k_{M'}$ correspond respectively to $n$ and $\kappa_{M(\pi)}$, whereas $M^*$ in \Cref{alg:lm:sublinear:extension} is renamed $M$.)
	
	\lmsublinearextension*
	
	\begin{proof}
	    By \Cref{sublinear-expectation}, \Cref{alg:lm:sublinear:extension} runs in expected $\tilde{O}_\epsilon(n)$ time and returns a correct solution with high probability. To improve its running time guarantee to a high probability bound we only need to incur a blowup of $O(\log(n))$ in running time: run the algorithm $O(\log(n))$ time in parallel and output the solution given by the first terminating copy. One of these algorithms will terminate within at most twice the expected time, by Markov's inequality, and so w.h.p., one of these completes after $\tilde{O}_{\eps}(n)$ time. Finally, by union bound and \Cref{correct-whp}, all of the $\log n$ algorithms' output satisfies the desired bounds with probability $1 - 1/poly(n)$, and so w.h.p., we obtain a solution satisfying the desired bounds after $\tilde{O}_{\eps}(n)$ time. 
    \end{proof}

	\section{Omitted Proofs from \Cref{sec:algo-general-streaming}}\label{appendix:general}
	
	Here we prove the tighter bound on the number of $V(M_1)$-disjoint $3$-augmenting paths in the subgraph $M_1\cup M_2$ as output by \Cref{alg:general-streaming}, restated below.
	\augtighter*
	\begin{proof}
		Fix a maximum set of disjoint length-three augmenting paths w.r.t.~$M_1$ in $G$, denoted by $\calP^*$. 
		By \Cref{poor-approx-maximal-matching=many-aug-paths}, we have $|\calP^*|\geq \left(\frac{1}{2}-3c-\frac{7\eps}{2}\right)\cdot \mu(G)$. 
		Next, let $S\subseteq \calP^*$ be the paths $u'-u-v-v'$ that ``survive'' the bipartition, in the sense that $(u,u'),(v,v')\in E_2$.
		By construction, each path in $\calP^*$ survives with probability exactly $\frac{1}{4}$. Therefore, $\E[|S|]\geq \frac{1}{4}\cdot \left(\frac{1}{2}-3c-\frac{7\eps}{2}\right)\cdot \mu(G)$. 
		Let $D:=\calP^*\setminus S$ be the set of paths that did not survive this bipartition.
		
		For $i \in \{0, 1, 2\}$, let $S_i\subseteq S$ and $D_i\subseteq D$ be the sets of paths $u'-u-v-v'$ in $S$ and $D$ (respectively) with $i$ of their $V(M_1)$ nodes $u$ and $v$ matched in $M_2$.
		Now, by our bipartition, if $u'-u-v-v'\in S_2\cup D_2$, i.e., if $u$ and $v$ are both matched in $M_2$, then they are matched to distinct nodes. 
		Therefore, $G[M_1\cup M_2]$ contains a set of augmenting paths $\calP$ w.r.t.~$M_1$ that are disjoint in their $V(M_1)$ nodes, of cardinality $|\calP|=|S_2|+|D_2|$. We now turn to lower bounding $|S_2| + |D_2|$.
		
		To bound $|S_2|+|D_2|$, we will double count the edges of $M_2$, once from their $V(M_1)$ endpoints, and once from their $\overline{V(M_1)}$ endpoints.
		First, by definition, since each edge in $M_2$ has exactly one endpoint in $V(M_1)$ and each node in $V(M_1)$ is matched at most once in $M_2$, we have that $|M_2|= 2|S_2|+|S_1|+2|D_2|+|D_1|\leq |S_2|+|D_2| + |M_1|,$ where the inequality follows from $|M_1| \geq \sum_{i=0}^2 (|S_i|+|D_i|)$, by definition.
		On the other hand, for each of the $|S|-|S_2|=|S_0|+|S_1|$ survived paths $u'-u-v-v'\in S_0\cup S_1$ that does not have both its internal nodes matched in $M_2$, we have by maximality of $M_2$ that $u'$ and/or $v'$ must contribute $b$ distinct edges to $M_2$. Therefore, $|M_2|\geq b\cdot (|S_0|+|S_1|)$. Combining the above, we obtain 
		\begin{align*}
		    b\cdot (|S|-|S_2|) \leq |M_2| \leq |S_2|+|D_2| + |M_1|,
		\end{align*}
		which after rearranging, yields
		\begin{align*}
		    b\cdot |S| - |M_1| \leq (b+1)\cdot |S_2|+|D_2|\leq (b+1)\cdot (|S_2|+|D_2|).
		\end{align*}
		Simplifying and combining with the lower bound on $\E[|S|]$, we obtain the claimed bound, as follows.
		\begin{align*}
		    \E[|\calP|]=\E[|S_2|+|D_2|] & \geq \frac{b}{b+1}\cdot \left(\E[|S|] - \frac{1}{b}\cdot |M_1|\right) \\
		    & \geq \frac{b}{b+1}\cdot \left(\frac{1}{4}\cdot \left(\frac{1}{2}-3c-\frac{7\eps}{2}\right) - \frac{1}{b}\cdot \left(\frac{1}{2}+c\right)\right)\cdot \mu(G). \\
		    & = \frac{b}{b+1}\cdot \left(\frac{1}{4}\cdot \left(\frac{1}{2}-3c\right) - \frac{1}{b}\cdot \left(\frac{1}{2}+c\right) -\frac{7\eps}{8} \right)\cdot \mu(G).
		    \qedhere 
		\end{align*}
    \end{proof}

\section{Omitted Proofs of \Cref{sec:adaptive}}\label{appendix:adaptive}

	We stress that the following is essentially implied by the work of \cite{wajc2020rounding}, from which we now repeat significant amount of text essentially verbatim.
	The only difference here will be our final proof of \Cref{periodic-kernel}, allowing us to efficiently periodically compute an $\eps$-AMM, and the use of this lemma in the subsequent section.
	Readers familiar with \cite{wajc2020rounding} are encourage to read ahead to that lemma.

	\paragraph{Overview.} Briefly, \cite{wajc2020rounding} identified an edge-coloring-based approach to compute, based on the efficient maintenance of edge colorings
	and a particular fractional matching of 
	\cite{bhattacharya2018deterministic},
	a kernel. (See \Cref{alg:sparsify}.)
	We start by recalling the type of fractional matching needed here, due to  \cite{arar2018dynamic}.
	
	\begin{Def}\label{def:approx-max}
	For $c\geq 1$ and $d\geq 1$, a fractional matching $\vec{x}$ is \emph{$(c,d)$-approximately-maximal $(c,d)$-AMfM} if  every edge $e\in E$ either has fractional value $x_e > 1/d$ or it has one endpoint $v$ with $\sum_{e\ni v} x_e \geq 1/c$ with all edges $e'$ incident on this $v$ having value $x_{e'} \leq  1/d$.
    \end{Def}
	
	As proven in \cite[Appendix A]{arar2018dynamic}, the dynamic fractional matching of \cite{bhattacharya2017fully} is precisely such an approximately-maximal matching.
	
	\begin{restatable}{lem}{fractwopolylog}\label{lem:two-polylog}
	    For all $\epsilon\leq \frac{1}{2}$, there is a deterministic dynamic $(1+2\epsilon, \max\{54\log n/\epsilon^3,(3/\epsilon)^{21}\})$-AMfM algorithm with  $t_u=O(\log^3n/\epsilon^7)$ worst-case update time, changing at most $O(\log n/\epsilon^2)$ edges' fractions per update in the worst case.
    \end{restatable}
	
	Now, we turn to the sparsification procedure of \cite{wajc2020rounding}, given in \Cref{alg:sparsify}.
	Briefly, this algorithm decomposes the graph into a logarithmic number of subgraphs, based on grouped $x$-values, edge colors these subgraphs using at most $\gamma=2$ times their maximum degree, and then outputs the union of these subgraphs.

	\begin{algorithm}[h] 
	\caption{Edge-Color and Sparsify \cite{wajc2020rounding}}
	\label{alg:sparsify}
	\begin{algorithmic}[1]
		\smallskip
		\For{$i\in \{1,2,\dots, \lceil 2\log_{1+\eps} (n/\epsilon)\rceil\}$}
		\State let $E_i \triangleq \{e \mid x_e\in ((1+\eps)^{-i},(1+\eps)^{-i+1}]\}$.
		\State compute a $2 \lceil (1+\epsilon)^i\rceil$-edge-coloring $\chi_i$ of $G_i\triangleq G[E_i]$. \Comment{Note: $\Delta(G_i) < (1+\epsilon)^i$}
		\State Let $S_i$ be a sample of $\min\{2 \lceil d(1+\epsilon)\rceil ,2 \lceil(1+\epsilon)^i\rceil\}$ colors without replacement in $\chi_i$. \label{line:sample-colors} 
		\EndFor
		\State \textbf{Return} $K\triangleq (V,\bigcup_i \bigcup_{M \in S_i} M)$.
	\end{algorithmic}
    \end{algorithm}

	The following lemma of \cite{wajc2020rounding} allows us to compute kernels from AMfMs using \Cref{alg:sparsify}. 
	\begin{restatable}{lem}{integralsparsifierwhp}\label{kernel-whp}
	Let $c\geq 1$, $\epsilon>0$ and $d\geq \frac{9c(1+\epsilon)^2\cdot \log n}{\epsilon^2}$.
	If $\vec{x}$ is a $(c,d)$-AMfM, then the subgraph $K$ output by \Cref{alg:sparsify} when run on $\vec{x}$ with $\epsilon$ and $d$ is a $(c(1+O(\epsilon),d(1+O(\epsilon),0)$-kernel, w.h.p.
\end{restatable}

    We are now ready to prove our (periodic) algorithmic kernel and AMM algorithm's guarantees, restated below for ease of reference.
	
    \periodickernel*
	\begin{proof}
	    We maintain the dynamic $(1+2\eps,\tilde{O}_{\eps}(1))$-AMfM of \Cref{lem:two-polylog}, using $\tilde{O}_{\eps}(1)$ deterministic w.c.~update time and number of changes to edges per update.
	    In addition, we maintain the subgraphs $G_i$ in \Cref{alg:sparsify}. In each such subgraph we maintain $2\lceil(1+\eps)^i\rceil$-color edge colorings in each $G_i$ in $O(\log n)$ deterministic w.c.~time per change to $\vec{x}$, using the logarithmic-time $(2\Delta-1)$-edge coloring algorithm of \cite{bhattacharya2020deterministic}.
	    This concludes the description of the updates, which by the above take deterministic w.c.~update time $t_u=\tilde{O}_{\eps}(1)$.
	    
	    Next, to compute a kernel, we run the sampling step of \Cref{alg:sparsify}. As this is bottlenecked by the time to write down the $O(\log^2 n)$ colors (matchings), each of size no greater than $\mu(G)$ (by definition), this query takes deterministic $\tilde{O}(\mu(G))$. 
	    Finally, this output graph $K$ is an $(O(\eps),d(1+O(\eps))$-kernel w.h.p., by \Cref{kernel-whp}.
	    Finally, to output an $\eps$-AMM, we appeal to the static algorithm \Cref{static-AMM}, which runs in deterministic time $\tilde{O}_\eps(d\cdot \mu(G))=\tilde{O}(\mu(G))$ and outputs an $\eps$-AMM, provided $K$ is a kernel, i.e., it also succeeds w.h.p.
	\end{proof}
\bibliographystyle{alpha}
\bibliography{abb,ultimate}
	
\end{document}